\documentclass[runningheads]{llncs}
\usepackage{url}
\usepackage{listings}
\usepackage{booktabs}
\usepackage{array}
\usepackage{dcolumn}\newcolumntype{d}[1]{D{.}{.}{#1}}
\usepackage{longtable}
\usepackage{microtype}
\usepackage[T1]{fontenc}
\usepackage[utf8]{inputenc}
\usepackage{amsfonts}
\usepackage{amssymb}
\usepackage{mathtools}
\usepackage{prftree}
\usepackage{algorithmic}
\usepackage{algorithm}
\usepackage{scalefnt}
\usepackage{pdfpages}
\usepackage{stmaryrd}
\usepackage{orcidlink}
\definecolor{Crayola}{HTML}{F45B69}
\definecolor{Cornflower}{HTML}{788BFF}
\definecolor{Fawn}{HTML}{F7B267}
\definecolor{Celadon}{HTML}{98C9A3}
\definecolor{Midnight-green}{HTML}{13505B}

\newcommand{\class}[1]{{\relax\ifmmode\mathbf{#1}\else\textbf{#1}\fi}}
\usepackage{hyperref}
\usepackage{color}

\spnewtheorem{notation}{Notation}{\bfseries}{\itshape}
\newcommand{\appref}[1]{the full version of the paper~\cite{fullversion}}

\DeclareMathOperator{\Erf}{erf}
\DeclareMathOperator{\Erfc}{erfc}
\DeclarePairedDelimiterXPP\erf[1]{\Erf\mkern1mu}(){}{#1}
\DeclarePairedDelimiterXPP\erfc[1]{\Erfc\mkern1mu}(){}{#1}

\DeclareMathOperator{\lfp}{lfp}
\DeclareMathOperator{\gfp}{gfp}
\DeclareMathOperator{\var}{var}
\DeclareMathOperator{\fn}{fn}
\DeclareMathOperator{\Free}{free}
\DeclareMathOperator{\conv}{conv}
\DeclareMathOperator{\link}{lk}

\newcommand{\declarebig}[2]{%
	\expandafter \newcommand\csname Scale#1\endcsname[1]{\vcenter{\hbox{\scalefont{##1}#2}}}%
	\expandafter \DeclareMathOperator\expandafter*\csname big#1\endcsname{%
		\vphantom\sum%
		\mathchoice{\csname Scale#1\endcsname{2}}{\csname Scale#1\endcsname {1.4}}{\csname Scale#1\endcsname {1}}{\csname Scale#1\endcsname {0.75}}}%
}

\declarebig{forall}{$\forall$}
\declarebig{exists}{$\exists$}
\declarebig{expects}{$\mathbb{E}$}

\newcommand{\R}{\stackrel{\scriptstyle \$}{\leftarrow}}

\newcommand{\algrule}[1][.2pt]{\par\vskip.5\baselineskip\hrule height #1\vskip.5\baselineskip}
\newcommand{\procname}[1]{%
	\textup{\texttt{#1}}%
}

\renewcommand{\O}{\mathcal{O}}
\newcommand{\sem}[1]{\llbracket {#1} \rrbracket}
\newcommand{\gbcshort}{GBC}
\newcommand{\gbclong}{generalized boolean circuit}
\newcommand{\Solver}[1]{\textnormal{\textsf{#1}}\relax}
\newcommand{\checker}{\Solver{iSMC}}
\newcommand{\prot}{\procname{TraceCert}}
\newcommand{\protGC}{\procname{TraceCertRev}}
\newcommand{\lib}{\Solver{BL-IP}}
\newcommand{\blic}{\Solver{blic}}
\newcommand{\blictwo}{\Solver{clic}}

\newcommand{\F}{\mathbb{F}_p}
\newcommand{\B}{\mathbb{B}}
\newcommand{\Alg}[1]{\textnormal{\texttt{#1}}}
\newcommand{\CPCertify}{\texttt{CPCertify}}
\newcommand{\Claimset}{\mathcal{C}}
\newcommand{\Eval}[1]{\sem{{#1}}}
\newcommand{\Pev}[1]{{\mathrm{\Pi}}_{{#1}}}

\newcommand{\ApplyEBDD}{\Alg{ApplyEBDD}}
\newcommand{\ComputeEBDD}{\Alg{ComputeEBDD}}
\newcommand{\Apply}{\Alg{Apply}}
\newcommand{\Reduce}{\Alg{Reduce}}
\newcommand{\etal}{\textit{et al.\ }}
\newcommand{\Abs}[1]{\mathopen|#1|\mathclose}

\newcommand{\nusmvoverhead}{5.84}
\newcommand{\proveroverhead}{2.00}
\newcommand{\verifierspeedup}{33.4}
\newcommand{\verifiermaxtime}{3.6}
\newcommand{\improvementsspeedup}{2.26}
\newcommand{\improvementmemory}{16.12}
\newcommand{\improvementmemorymedian}{26.26}
\newcommand{\naivetimeout}{4}
\newcommand{\minprobability}{$2.54 \cdot 10^{-7}$}

\begin{document}
\title{{\checker}: A BDD-based Symbolic Model Checker with Interactive Certification}%
\titlerunning{A BDD-based Symbolic Model Checker with Interactive Certification}

\author{
Philipp Czerner \orcidlink{0000-0002-1786-9592} \and
Javier Esparza   \orcidlink{0000-0001-9862-4919} \and
Konrad Winslow  \orcidlink{0009-0001-2506-097X} 
}
\authorrunning{P.\ Czerner, J.\ Esparza, K.\ Winslow}

\institute{Technical University of Munich, Germany,\\
\email{\{czerner, esparza, mko\}@cit.tum.de}
}

\maketitle
\begin{abstract}
We present {\checker}, the first self-certifying model checker with interactive certification, a certification paradigm based on the theory of interactive proof systems. {\checker} is a symbolic BDD-based model checker for arbitrary properties of Computation Tree Logic (CTL) with justice requirements. After solving an instance of the model-checking problem,  {\checker} conducts a certification procedure that guarantees with high probability (chosen by the user) that the answer is correct. {\checker} is based on the technology of the QBF-solver with interactive certification presented by Couillard \etal at CAV 2023. We extend, improve on, and re-implement this technology, adapting it to the needs of CTL model checking. 
\end{abstract}


\section{Introduction}
\label{ch:1}

Interactive certification is a certification paradigm based on the theory of interactive
proof systems \cite{GoldwasserMR85,Arora2009}. In the standard certification paradigm 
an agent solves a problem (e.g. satisfiability of a Boolean formula), produces a certificate (e.g. a satisfying assignment or a resolution proof of unsatisfiability), and sends it to another agent, which checks the certificate and accepts it or not. Interactive certification generalizes this paradigm by allowing the agents to engage in a protocol with multiple rounds in which the verifier repeatedly and adaptively ask questions to the other agent.

We present {\checker}, the first self-certifying model checker with
interactive certification. {\checker} is a symbolic
BDD-based model checker for properties expressed in Computation Tree Logic (CTL)
\cite{Chaki2018,Bryant2018}. It accepts flattened boolean models in the
established \textsf{smv} format and arbitrary CTL formulas and justice
requirements \cite{McMillan93,Piterman2018}. After solving an instance of the
model-checking problem,  {\checker} conducts a certification procedure that
guarantees with high probability---chosen by the user---that the answer is
correct. {\checker} extends, improves, and reimplements technology of the
BDD-based QBF-solver developed by Couillard et al. in \cite{CAV23}.

\smallskip\noindent \textbf{Architecture and functionality of {\checker}.}
Conceptually, {\checker} consists of three modules: Solver, Prover and Verifier.
Solver and Prover are implemented on top of the {\blictwo} BDD-library, an
optimization and extension of the {\blic} library of \cite{CAV23}. Given a
system and a CTL specification, Solver first computes a BDD representing the set
of all reachable states of the system satisfying the specification; essentially,
Solver implements  
the same algorithm as SMV,  NuMSV, or (the BDD-based part of) NuXmv
\cite{Chaki2018,CavadaCDGMMMRT14,Cimatti2000,McMillan93}. Then, Prover and
Verifier engage in an \emph{interactive proof protocol} called {\prot}. Loosely
speaking, an interactive proof protocol specifies a sequence of interactions
between Prover and Verifier, starting with a claim by Prover about the answer to
a computational question, and ending with Verifier deciding to believe Prover or
not \cite{Arora2009}.  At the start of {\prot}, Prover sends Verifier the
\emph{execution trace} of Solver on the model-checking instance, consisting of
the sequence of calls to the BDD library executed by Solver, interspersed with
assertions describing the decisions made at conditional branches. The trace ends
with an assertion about the final result, stating either that the system satisfies the
specification or that it does not\footnote{The length of the execution trace can grow
exponentially in the size of the model-checking instance in the worst case, but it is usually much
smaller than the computation time, and can even be exponentially smaller.}.
After receiving the trace, Verifier asks Prover questions about it in such a way
that, after {\prot} terminates, Verifier can tell with high probability whether
all the assertions of the trace are true, \emph{without executing it}. 
For this, {\prot} encodes boolean functions as multivariate polynomials over a
finite field $\F$, where $p$ is some large prime. For example, the function $x_1
\vee x_2$  is encoded as the polynomial $x_1 + x_2 - x_1x_2$. Intuitively, the
polynomial behaves like the formula for $x_1, x_2 \in \{0,1\}$, but Verifier
asks questions about the values of the polynomials at points chosen uniformly at
random from $\F$. This guarantees that the error probability, defined as the
probability that Verifier does not catch a wrong result by Prover  or any
malicious attempt by Prover to  ``fool'' Verifier, is very small.

{\checker} exhibits three fundamental properties:
\begin{enumerate}
\item If at least one of the assertions of the execution trace is false, then
the error probability is at most $\ell/|\F|$. In our experiments we use $p=
2^{61} -1$, and the error probability never exceeds
\minprobability.\footnote{For smaller probabilities one can take a larger $p$ or run $\prot$ multiple times.}
\item Verifier runs in time $O(n^2 \ell)$, where $n$ is the number of variables
of the model-checking instance, and $\ell$ is the length of the execution trace.
(Verifier's runtime depends only on the length of the trace, not on the time it
takes to execute it. Since the trace consists of a sequence of BDD operations,
the time can be exponentially larger.) In particular, Verifier runs in
polynomial time in the number of variables for systems whose state space has
fixed diameter or for bounded model-checking problems.  In experiments conducted
with a timeout of 15 minutes for Prover, Verifier never needs more than \verifiermaxtime\
seconds.
\item For any model-checking instance, if Solver runs in time $S$, then Prover
runs in time $O(S)$. In our experiments, the constant hidden in the big-oh
notation lies between 1 and 2.8.
\end{enumerate}
Currently, the price to pay for interactive verification is a penalty in the
efficiency of Solver. In our experiments, conducted with a timeout of 15
minutes, the average slowdown factor w.r.t.\ NuSMV is \nusmvoverhead. The properties above
make {\checker} particularly attractive for architectures in which a client with
limited computational resources asks a powerful but \emph{untrusted} server to
solve a model-checking instance. In such an architecture, {\checker}'s Solver
and Prover run on the server, while Verifier runs on the client. By property 2.,
Verifier only invests linear time in $\ell$, and so it can certify even very
large instances of the model-checking problem. By property 3., this is achieved
with reasonable overhead. 

To the best of our knowledge, properties 2.-3. are a unique feature of
{\checker}. The reason is a fundamental theoretical limit:  to the best of our
knowledge, all certification procedures implemented in current model checkers
are non-interactive protocols in which Prover sends Verifier one or more
objects, called \emph{certificates}, and then, \emph{without further interaction
with Prover}, Verifier runs an algorithm on the certificate and emits a verdict.
It is well-known that such certificates have worst-case exponential length in
$\ell$  unless $\class{NP}=\class{coNP}$ (see e.g.~\cite{Arora2009}), and so
Verifier needs exponential time in $\ell$. 



\smallskip\noindent \textbf{Main technical contributions.}  Our first main
contribution is the observation that the problem of checking all assertions of
the execution trace can be reduced to the problem studied in \cite{CAV23}:
computing the number of satisfying assignments of a \emph{boolean circuit with
partial evaluation}, an extension of standard boolean circuits introduced in
\cite{CAV23}. (The reduction is sketched in the next section.) This allows one
to reuse {\CPCertify},  the interactive proof protocol of \cite{CAV23}, and its
implementation on top of the {\blic} BDD library, also developed in
\cite{CAV23}.  However, the resulting tool is inefficient, because neither
{\CPCertify} nor {\blic} are tailored to the needs of symbolic model checking.
Our three other main contributions are solutions to three bottlenecks of this
direct approach:

\begin{enumerate}
\item {\CPCertify} proceeds in rounds, one for each gate of the input boolean
circuit.  
Many of these gates are labeled with equivalence and renaming operators.
However, in {\CPCertify} the rounds for these gates are very expensive for
Prover:  the round for an equivalence gate $\psi_1\equiv \psi_2$ takes $\O(n_1
\cdot n_2)$ time, where $n_i$ is the BDD-size of $\psi_i$, and the round for a
renaming gate $[X'/X]\psi$ takes $\O(n^{2^k})$ time, where $n$ is the BDD-size
of $\psi$ and $k = |X|$. In {\prot}, our new protocol, these rounds take $\O(1)$
and $\O(n)$ time, respectively.
\item {\CPCertify} traverses the circuit in topological order, starting at its
output gate and moving towards its input gates. However, Solver proceeds in the
\emph{reverse} order, from inputs to outputs. For this reason, Solver needs to
store all BDDs for all gates of the circuit. This prevents the use of garbage
collection, an important feature of BDD-libraries for discarding BDD-nodes no
longer required by the application (see e.g. \cite{LvSX13,Somenzi2015}). We show
that, under the assumption that Prover acts as an oracle (meaning that it does
not store information from previous queries) our new protocol {\prot} can
traverse the circuit in the same order as Solver without runtime penalty. 
\item BDD-libraries use a global \emph{computation cache} of BDD-nodes for
\emph{all} the boolean functions computed along an execution trace and their
sub-functions (see e.g. \cite{Somenzi2015}). For model-checking applications,
the table leads to efficiency gains of 1-2 orders of magnitude \cite{Yang1998}.
However, {\blic} does not use a global table. The reason is that {\blic}
manipulates not only BDDs, but \emph{extended BDDs} (eBDDs), a data structure
introduced in \cite{CAV23}, and {\blic}'s implementation of eBDD operations is
incompatible with a global computation cache.  We introduce a novel representation of eBDDs
that solves this problem.
\end{enumerate}
In our experimental comparison, the
contributions 1.-3. lead to an average twofold reduction in execution time, where the reduction factor increases with the execution time and reaches a maximum of 73, and an average sixteen fold reduction in memory usage.

\smallskip\noindent\textbf{Related Work.}  
Namjoshi introduced a certification procedure for $\mu$-calculus model checking
based on deductive proof systems \cite{Namjoshi2001}. Griggio \etal also propose
to use deductive systems for certification of  LTL model checking
\cite{GriggioRT21}. We follow a different approach that does not require to use
deductive systems. Yu \etal have developed a certification procedure for
SAT-based model checking  that uses inductive invariants as certificates for
$k$-induction \cite{YuBH20,YuFBH22,YuFBH23,FroleyksYBH24}. Jussila \etal
presents a method to generate proof certificates from BDDs that can also be used
to construct certificates for symbolic BBD-based model-checking of safety
properties \cite{Jussila2006}. Conchon \etal and Mebsout and Tinelli construct
certificates for SMT-based model checking of safety properties of infinite-state
and parameterized systems \cite{ConchonMZ15,MebsoutT16}.  All these approaches are limited to safety properties, while we target arbitrary CTL properties with justice requirements. Kuismin and Heljanko present a certification procedure for LTL liveness properties that works by reduction to certification of safety
properties \cite{KuisminH13}.

All the approaches above generate certificates of worst-case exponential size
in the size of the instance, even for bounded model-checking problems, and so 
their Verifier components need exponential time and space in the size of the instance. In our approach Verifier only needs polynomial time. 

Interactive certification for QBF-solving and SAT-solving using interactive
proof systems has been studied in \cite{CAV23,CzernerEK24}. There is also
recent interest  in zero-knowledge, interactive proof systems for unsatisfiable
SAT formulas \cite{Luo2022}. The approach has been recently extended to
zero-knowledge proofs for all PSPACE problems \cite{karthikeyan2025}. 

Our approach follows the paradigm of certifying computations without reexecuting them \cite{WalfishB15}. We focus on computations consisting of calls to a BDD-library, which allows us to obtain a certification procedure with much smaller overhead.

Verified model checkers are an alternative to certification
\cite{Sprenger98,EsparzaLNNSS13,Wimmer19,TsaiFLSWY23a}. They do not need to
produce or check certificates, but require to maintain the correctness proof
whenever the implementation of the model checker changes.

\smallskip\noindent\textbf{Structure of the paper.} Section~\ref{ch:2} fixes some
notation on CTL model checking and introduces interactive proof protocols.
Section~\ref{sec:checker} describes the structure of  {\checker}. Section
\ref{sec:prot} presents {\prot}, the extension of the interactive proof protocol
of \cite{CAV23} used by {\checker}. Section \ref{sc:prover} presents our
implementation of the Prover of \ref{sec:prot}. Section \ref{sec:eval} presents
our experimental results.  Throughout the paper we refer the reader to several appendices
containing formal definitions and proofs.

\section{Preliminaries}
\label{ch:2}

We assume the reader is familiar with computation tree logic (CTL), the
bottom-up  model-checking algorithm for CTL, and its symbolic implementation
with BDDs \cite{Clarke2018,Piterman2018,Emerson1986}. We briefly recall a few
notions.

Given a Kripke structure with set of states $S$ and transition relation $R
\subseteq S \times S$ and a CTL  formula $\varphi$, the algorithm computes the
set $\sem{\varphi} \subseteq S$ of states satisfying $\varphi$ in bottom-up
manner. For example, for the formula $\mathbf{EG} \mathbf{EF} p$ the algorithm
first computes $\sem{p}$; then it computes $\sem{\mathbf{EF} p}$ using the
identity  $\sem{\mathbf{EF} p} = \lfp_Z (R^{-1}(Z) \cup \sem{p})$, where $\lfp$
denotes the least fixpoint of the mapping $Z \mapsto R^{-1}(Z) \cup \sem{p}$;
finally, it  computes $\sem{\mathbf{EG} \mathbf{EF} p}$ using the identity
$\sem{\mathbf{EG} \mathbf{EF} p} = \gfp_Z (R^{-1}(Z) \cap \sem{\mathbf{EF} p})$,
where $\gfp_Z$ denotes the greatest fixpoint of $Z \mapsto R^{-1}(Z) \cup
\sem{\mathbf{EF} p})$. 

Symbolic model checking encodes a state of the Kripke structure as a valuation
of a set $X$ of boolean variables, a set of states $T \subseteq S$ as a boolean
function $T^b(X)$ over $X$, and the transition relation $R$ as a boolean
function $R^b(X, X')$ over the variables $X \cup X'$, where $X' = \{v' \mid v
\in X\}$ is a second set of primed variables. In particular, we have $(S \cap
D)^b(X) = S^b(X) \wedge D^b(X) $, $(S \setminus D)^b(X)  = S^b(X) \wedge \neg
D^b(X)$, and $(R^{-1}(S))^b(X)  = \exists X'\  R^b(X,X') \wedge S^b(X)[X'/X]$,
where $S^b(X)[X'/X]$ denotes the result of substituting $x'$ for $x$ in $S^b(X)$ for
every variable $x \in X$. 

BDD-based symbolic model checkers for CTL, like NuSMV~\cite{Cimatti2000}, represent and manipulate
boolean functions as (ordered and reduced) binary decision diagrams
(BDDs)~\cite{Bryant1986,Bryant2018}. Table~\ref{table:bddops} shows the
BDD-based boolean function library interface of our library {\blictwo}, with primitive operations above the line and derived operations below it. More details on these operations are given in Appendix \ref{apx:bdds}.

\begin{table}[ht]
  \centering
  \caption{BDD operations of \blictwo.}
  \begin{tabular}{llll}
    \toprule
    Name                      & Boolean formula                                          & Complexity             & Implementation                                                            \\
    \midrule Binary Operation & $f \circledast g$                                        & $O(|f| \cdot |g|)$     & $\texttt{Apply($f, g, \circledast$)}$                                     \\
    Negation                  & $\neg f$                                                 & $O(|f|)$               & $\texttt{Apply($f, 1, \overline{\wedge}$)}$                               \\
    Restriction               & $f|_{x_i \leftarrow b \in \{0, 1\}}$                     & $O(|f|)$               & $\texttt{Restrict($f, v_i, b$)}$                                          \\
    Renaming                  & $[x_i/{x_i}']f$                                          & $O(|f|)$               & $\texttt{Rename($f, x_i, {v_i}'$)}$                                       \\
    Equivalence Check            & $f \equiv g$                                    & $O(1)$                 & Root comparison                                                           \\
    Counting Solutions        & $\left|\{\overline{v} \mid f(\overline{v}) = 1\}\right|$ & $O(|f|)$               & Counting paths to $1$                                                     \\
    Evaluation                & $f(x_0,\dots,x_n)$                                       & $O(n)$                 & Graph traversal                                                           \\  \hline
    Composition               & $f|_{x_i \leftarrow g}$                                  & $O(|f|^2 \cdot |g|^2)$ & $g \wedge f|_{v_i \leftarrow 1} \vee \neg g \wedge f|_{v_i \leftarrow 0}$ \\
    Exists                    & $\exists_{x_i}f$                                         & $O(|f|^2)$             & $f|_{x_i \leftarrow 0} \vee f|_{v_i \leftarrow 1}$                        \\
    Forall                    & $\forall_{x_i}f$                                         & $O(|f|^2)$             & $f|_{x_i \leftarrow 0} \wedge f|_{x_i \leftarrow 1}$                      \\
    \bottomrule
  \end{tabular}
  \label{table:bddops}
\end{table}


\noindent \textit{Interactive Proof Protocols.}
Consider a computational problem consisting of, given a function $f$, computing
$f(x)$ for an input $x$. An \emph{interactive protocol} for $f$ is a
communication protocol  between two algorithms, called a \emph{Prover} and a
\emph{Verifier}. The protocol is run after Prover claims to Verifier that
$f(x)=d$ holds for some value $d$. The protocol tells Verifier how to choose a sequence of
questions for Prover, depending on $x$ and on Prover's answers to previous
questions, and how to decide whether to believe Prover's claim or not. As a
simple example, assume Prover claims to Verifier that two graphs $G_1, G_2$ are
\emph{not} isomorphic. The interactive protocol instructs Verifier to repeatedly
pick $G \in \{G_1, G_2\}$ uniformly at random, pick a permutation $G_\sigma$ of
$G$, again uniformly and random, and ask Prover which of $G_1$ and $G_2$ is
isomorphic to $G_\sigma$. After a fixed number $k$ of rounds, 
Verifier believes Prover's claim if{}f all of Prover's answers are correct.
It is easy to see that for any Prover algorithm, the probability that Prover
fools Verifier (i.e., that $G_1$ and $G_2$ are isomorphic, but Verifier believes
Prover's claim) is at most $2^{-k}$.

The complexity class $\class{IP}$~\cite{Arora2009} contains all decision
problems for which there exists a probabilistic Verifier (i.e., an algorithm that
can flip coins) such that (a) Verifier runs in polynomial time; (b) for every
input $x$, Verifier believes the \emph{honest Prover} that claims the correct
value for $f(x)$ and answers questions truthfully; and (c) for every input $x$
and for every Prover claiming a wrong value for $f(x)$, Verifier believes Prover
with probability at most $2^{-|x|}$.

It is easy to see that $\text{SAT}$ is in $\class{IP}$, which implies
$\class{NP} \subseteq \class{IP}$: the 1-round protocol where Prover sends
Verifier an assignment, and Verifier checks in polynomial time that it satisfies
the formula, shows it.  It is harder to show $\text{UNSAT}  \in \class{IP}$:  the
1-round protocols in which Prover sends Verifier the complete truth table of the
formula or, say, a resolution proof do not work, because Verifier cannot check
them (not even read them!) in polynomial time \emph{in the size of the formula}
\footnote{Verifier can check them in polynomial time in the size of the truth
table or the resolution proof, but that is not what the class $\class{IP}$ specifies.}. The same happens
for certificates in any other proof system unless $\class{NP} =  \class{coNP}$.
Finally, Shamir's celebrated $\class{IP} = \class{PSPACE}$ theorem shows that
every problem in \class{PSPACE} has an interactive protocol with a
polynomial-time Verifier \cite{Shamir1992,Lund1992}. 


\section{Introduction to {\checker}} 
\label{sec:checker}
We explain the
algorithm underlying {\checker} by means of an example . Consider the following scenario. Let $I(X)$ and  $R(X, X')$ be
boolean functions representing the set of initial states and the transition
relation of a system, and let $p(X)$ be a function representing the set of
states satisfying an atomic proposition $p$.  We run {\checker}'s Solver to
decide if all initial states satisfy $\mathbf{EF} p$. For simplicity, we assume that Solver computes
the set of all predecessors of $p(X)$ (i.e., all states from which it is possible to reach some state satisfying $p$ in arbitrarily many steps) and then intersects it with $I(X)$, as
shown at the top of figure \ref{fig:inf} on the left. We abbreviate $I(X), R(X,X'), C(X) \ldots$ to  $I, R, C \ldots$, and let $\textbf{Pre}_C(X) :=\exists X' \colon (R(X, X') \wedge C(X) [X'/X])$
denote the set of immediate predecessors of a give set $C$ of states.

\begin{figure}[htb]
\begin{center}
{\footnotesize
\begin{minipage}[c]{0.6\textwidth}
\begin{minipage}[t]{0.45\textwidth}
  \begin{algorithmic}
    \STATE $C \leftarrow p$
    \REPEAT 
    \STATE $\textit{OldC} \leftarrow C$
    \STATE $C \leftarrow C \vee \textit{PreC}$
    \UNTIL $C \equiv OldC$
    \IF{$C \land I \equiv \textit{false}$} \RETURN $\textbf{false}$ 
    \ELSE \RETURN $\textbf{true}$
    \ENDIF
  \end{algorithmic}
\end{minipage}
\begin{minipage}[t]{0.45\textwidth}
\newcommand{\pho}{}
  \begin{algorithmic}
    \STATE \pho1: $C_0 \leftarrow p$
    \STATE \pho2: $\textit{PreC}_0 \leftarrow \textbf{Pre}_{C_0}$
    \STATE \pho3: $C_1 \leftarrow C_0 \vee \textit{PreC}_0$
    \STATE \pho4: \textcolor{magenta}{assert $C_1 \not\equiv C_0$}
    \STATE \pho5: $\textit{PreC}_1 \leftarrow \textbf{Pre}_{C_1}$
    \STATE \pho6: $C_2 \leftarrow C_1 \vee \textit{PreC}_1$
    \STATE \pho7: \textcolor{magenta}{assert $C_2 \equiv C_1$}
    \STATE \pho8: $C_3:= C_2 \land I$
    \STATE \pho9: \textcolor{magenta}{assert $C_3 \equiv \textit{false}$}
  \end{algorithmic}
\end{minipage}
\end{minipage}

\begin{minipage}[c]{0.7\textwidth}
  \vspace{0pt}
  \centering
  \includegraphics[width=\textwidth]{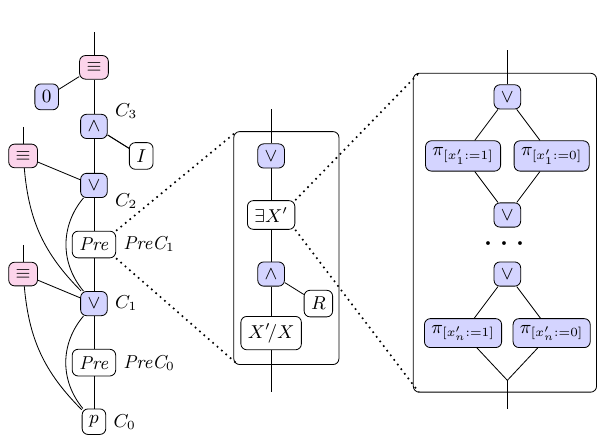}
\end{minipage}
}
\end{center}
\caption{At the top, symbolic checking algorithm for the formula $\textbf{EF} p$ (on the left),  and one of its
possible execution traces (on the right). At the bottom, circuit for this execution. Blue nodes correspond to basic gates, white nodes to ``macros''. }
\label{fig:inf}
\end{figure}

Assume further that Solver exits the repeat loop after two iterations and returns
\textit{true}. The corresponding execution trace is shown at the top of figure \ref{fig:inf}, on the right. 
The trace is transformed into a \emph{generalized boolean circuit}, a generalization of the 
boolean circuits with partial evaluation used in \cite{CAV23} defined below, and a set of \emph{equivalence claims}. In a nutshell, every instruction of the
trace is transformed into a gate of the circuit, and each assertion is translated into an equivalence claim. In particular, this reduces the problem
of certifying that the execution trace is correct is reduced to a generalization of the problem solved in \cite{CAV23},
for which we then give an interactive protocol. 

\medskip \noindent \textit{Generalized boolean circuits (GBCs).}  Let $X$ be a set
of boolean variables.  GBC-expressions over $V$ are given by the grammar
$$\varphi ::= \top \mid  \bot \mid x \mid \neg \varphi \mid \varphi
  \circledast \psi \mid \pi_{[x:=b]}\varphi  \mid \varphi[y/x]$$ \noindent
where $\circledast$ is any boolean operator, $x,y \in X$, $b \in \{0,1\}$, and $y$ does not appear in $\varphi$.
The GBC of a GBC-expression
$\varphi$ is the acyclic graph with one node for each subexpression of $\varphi$
and edges corresponding to the immediate subexpression relation. In analogy to boolean circuits, we call the nodes of a GBC \emph{gates}.
GBCs extend the circuits used in \cite{CAV23} with renaming gates $\varphi[x'/x]$.

The semantics of a GBC $\varphi$ is a boolean function $\fn(\varphi)$ over $X$,
defined by $\fn(\top):=1$, $\fn(\bot):=0$, ${\fn(x):= x}$, $\fn(\varphi
\circledast \psi) = \fn(\varphi) \circledast \fn(\psi)$, $\pi_{[x := b]}\varphi
= \fn(\varphi)|_{x \leftarrow b}$ (set $x$ to $b$ in $\fn(\varphi)$), and
$\fn(\varphi[y/x]) := \neg x'\wedge \fn(\varphi)|_{x \leftarrow 0} \vee x'
\wedge \fn(\varphi)|_{x_ \leftarrow 1}$ (substitute $y$ for $x$ in
$\fn(\varphi)$). We denote by $\Free(\varphi)$ the set of variable gates
reachable from $\varphi$'s root.

\smallskip\noindent \textit{From traces to GBCs.} We illustrate the translation from execution traces to GBCs and equivalence claims by example.
The GBC for the  execution trace at the top of figure \ref{fig:inf} is shown at the bottom of the figure. The GBC corresponds to the blue and white nodes, where the blue nodes model basic gates and the white nodes 
are ``macros'', standing themselves for a circuit. For example, the node on the left labeled by \textit{Pre} is a ``macro'' for the circuit shown in the middle (recall that $\textit{Pre}_C(X) :=\exists X' \colon (R(X, X') \wedge C(X)[X'/X]$), and the node in the middle labeled by $\exists X'$ is a macro for the circuit on the right; we use the equivalence $\exists x \colon \varphi \equiv \pi_{x:=1}(\varphi) \vee \pi_{x:=0} (\varphi)$, valid for every formula $\varphi$. The nodes labeled by $p, I, R$ are macros for circuits for the boolean formulas $p(X), I(X), R(X,X')$. Observe that the size of the circuit is linear in the length of the trace
and of the inputs $I$, $R$, and $p$.   The pink nodes are just a graphical representation of the equivalence claims corresponding to the assertions of the trace.


\section{{\prot}: An interactive proof protocol for GBCs}
\label{sec:prot}

We describe {\prot}, our improvement on {\CPCertify}. {\prot} adds support for
efficient variable renaming and equivalence assertions, which are common in
execution traces. It takes as input the circuit $\varphi$ obtained from the
execution trace of Solver. A first preprocessing phase, already present in
{\CPCertify}, transforms $\varphi$ into a GBC with additional so-called degree reduction nodes, denoted
$\conv(\varphi)$.
We define $\conv(\varphi)$ in Section \ref{subsec:preproc}. $\prot$ itself is
presented in Section \ref{subsec:prot}.

\subsection{Preprocessing: From $\varphi$ to $\conv(\varphi)$}
\label{subsec:preproc}

\noindent \textit{Arithmetization of GBCs.} Arithmetization is a core concept
for interactive proof systems and underlies the proof of $\class{IP} =
\class{PSPACE}$~\cite{Lund1992}. The arithmetization of a boolean function $f$
is a polynomial $\Eval{f}$ over a finite field $\F$ s.t.\ $f(\sigma) =
\Eval{f}(\sigma)$ for every \emph{boolean} assignment
$\sigma$~\cite[Prop.~1]{CAV23}. This is achieved by defining
$\Eval{0}:=0, \Eval{1}:=1, \Eval{\neg f} := 1 - \Eval{f}$, $\Eval{f \land g} :=
\Eval{f} \cdot \Eval{g}$ and $\Eval{f \lor g} := \Eval{f} + \Eval{g} - \Eval{f}
\cdot \Eval{g}$. The arithmetization of a GBC $\varphi$, denoted
$\Eval{\varphi}$, is defined as the arithmetization of its associated boolean
function. For example, the arithmetization of the GBC $(x_4 \vee \pi_{[x_3 :=
0]} x_3) \wedge \pi_{[x_3 := 1]}x_3$ is $(x_4 + 0 - x_4 \cdot 0) \cdot 1 = x_4$
(For more details see Appendix \ref{apx:preprocess}). Arithmetization allows to
design interactive protocols in which Verifier asks Prover to partially evaluate
polynomials of circuits on  \emph{non-boolean} assignments, which allows her to
detect cheating Provers with good probability. We define partial evaluation
$\pi_{[x:=a]} p$ of polynomials, which replaces $x$ by $a$ in p. The notation
$\Pev\sigma p$ stands for partially evaluating p on each assignment $\sigma(x_i)
= a_i$.

\smallskip \noindent \textit{Degree reductions and the circuit $\conv(\varphi)$.} It is easy to see
that the degree of the polynomial $\Eval{\varphi}$ may be exponential in the
height of $\varphi$\footnote{Observe that, for example, $\Eval{\psi_1 \wedge \psi_2}=\Eval{\psi_1} \cdot \Eval{\psi_2}$,
and so the degree of $\Eval{\psi_1 \wedge \psi_2}$ can be twice the degree of $\Eval{\psi_1}$ and $\Eval{\psi_2}$.}
The design of interactive proof systems requires to limit
the degree, because Verifier otherwise will not run in polynomial time. As in
\cite{CAV23}, for every variable $x_k$ we introduce a
\emph{degree reduction operator} $\delta_{x_k}$ that reduces the exponents of all
powers of $x_k$ to $1$, e.g., $\delta_{x_2}(x_1x_2^3 - 2 x_1x_2^2 + 4) = x_1 v_2 - 2
  x_1x_2 + 4 = -x_1x_2+4$; observe that degree-reducing a polynomial does not
change its value under boolean assignments, and so a polynomial and its degree
reduction encode the same boolean function.

Since we encode boolean functions as polynomials, we can interpret GBCs as
arithmetic circuits over the field $\F$. For example, $\land$-gates correspond
to multiplication. To limit the degree of $\Eval{\varphi}$, we replace each
boolean gate $\psi = \psi_1 \circledast \psi_2$ of $\varphi$ by a circuit
$\psi':=\delta_{x_1}(\delta_{x_2 }( \cdots  \delta_{n}(\psi) \cdots))$, where $\delta_{x_k}$ is a
\emph{degree-reduction gate} for the variable $x_k$, with semantics
$\Eval{\delta_{x_k}(\varphi)} := \delta_{x_k}(\Eval{\varphi})$. We denote the final
result by $\conv(\varphi)$
It is easy to see that for every circuit $\varphi$ the polynomial
$\Eval{\conv(\varphi)}$ is unique and multilinear \cite[Prop. 3]{CAV23}; loosely speaking, the degree-reduction
nodes reduce all exponents of $\Eval{\varphi}$ to $1$.

\subsection{$\prot$: High-level view.}
\label{subsec:prot}

We introduce the protocol $\prot$. In this section we  give a high level view, and in the next
describe the elements of the protocol in more detail. Pseudocode for the
different parts of the protocol and correctness proofs can be found in
Appendix~\ref{apx:prot}, with the top-level procedure being described in
Appendix~\ref{subsec:ebcertifytoplevelapx}.

%

\smallskip\noindent \textit{Initialisation and goal.} Given a GBC $\conv(\varphi)$ over a
set $X$ of $n$ variables, obtained from an execution trace, Prover and Verifier
initiate $\prot$ by choosing a prime number $p \geq 2^{|X|}$. From now on, all
polynomials they exchange are over the finite field  $\F$. Prover starts by
making a set of \emph{ claims} about 
$\conv(\varphi)$ (Sec.~\ref{subsc:claimsgenerationapx} of Appendix \ref{apx:preprocess}). 
For GBCs coming from execution traces, these are equivalence claims corresponding
to the assertions of the trace; in other words, Prover is claiming that the execution
trace indeed satisfies the assertions. For example, the claims for the execution
trace of the program on the left of figure \ref{fig:inf} state that the while
loop was exited after two iterations (which implicitly means that $C_2$ is the
set of all predecessors of $p$), and that the set of states $(\neg C_2 \land I)(X)$ is empty.

Prover's initial claims are the initial content of a set $\Claimset$ of
claims that gets repeatedly updated throughout {\prot}. 
Some updates are deterministic, while others are probabilistic, i.e., the claims replacing a given one are sampled from a certain set. Both kinds of
updates maintain \emph{claim equivalence with high probability}, or CEHP,
defined as follows:  if all claims of $\Claimset$  are true
before an update, then all claims after the update are true, and if some claim of
$\Claimset$ is false before an update,  then, with probability $1-(k/|\F|)$ for
$k \ll |\F|$, some claim of $\Claimset$ is false after the update. At the end of $\prot$,
$\Claimset$ contains final claims that Verifier can check in polynomial time
(which was not true of the initial claims).  Verifier checks the final claims
and, if all are true, declares to believe Prover's initial claims.

\smallskip\noindent \textit{Round for a gate.} $\prot$ iterates over the gates of
$\conv(\varphi)$ in topological order, starting at the output gates and handling
each non-input gate at most once. At each non-input gate, Prover and Verifier
engage in a \emph{round}. Verifier starts the round for a gate $\psi$ by
collecting all claims of $\Claimset$ about $\psi$, and replacing them by a
single \emph{principal claim}. Then, Verifier (using the principal claim) sends
Prover a set of \emph{challenges}: questions about the immediate successors of
$\psi$ according to the topological order. Prover answers the challenges with
polynomials for each of $\psi$'s successors. Verifier conducts some consistency
checks on the answers. If the polynomials do not pass the checks, Verifier
declares it does not believe Prover and {\prot} terminates; otherwise, Verifier
replaces the principal claim (about $\psi$) by the new claims (about $\psi$'s
successors) derived using Prover's answers.

\smallskip\noindent \textit{Decision.} After the rounds, Prover's initial equivalence claims---which Verifier cannot directly check because their polynomials have worst-case exponential size---have been replaced by claims
about atomic circuit expressions of the form $\top, \bot, x_k$, whose polynomials are $0$, $1$, or $x_k$. Verifier checks these claims herself in linear time, and believes Prover's initial claims if{}f they are all true.

The next sections provide details. The syntax and semantics of Prover's claims
and Verifier's challenges, as well as a  \emph{claim normalization} step,
are described in Section \ref{ssubsec:clch}. The round for a gate is described in Section
\ref{subsubsec:round}.

\subsection{Detailed view of $\prot$: Claims, challenges, and claim normalization.}
\label{ssubsec:clch}

During the execution of $\prot$ Prover can make claims of the following types, where $\psi, \psi_1, \psi_2$ denote gates of $\conv(\varphi)$:
\begin{itemize}
  \item \emph{$\F$-Evaluation}: An $\F$-evaluation claim $\Pev{\sigma}(\Eval{\psi}) =
          k$, where $\sigma \colon X \rightarrow \F$ and $k \in \F$,
        states that evaluating the polynomial $\Eval{\psi}$ on $\sigma$ yields $k$.
  \item \emph{$\B$-Evaluation}: An evaluation claim $\Pev\sigma \fn(\psi) = b$, where
        $\sigma \colon X \rightarrow \{0, 1\}$ and $b \in \{0, 1\}$, states that evaluating the boolean function  $\fn(\psi)$ on $\sigma$ yields $b$.
  \item \emph{Count}: A count claim $\Sigma \fn(\psi) = k$ states that
        the boolean function $\fn(\psi)$ has exactly $k$ satisfying assignments.
  \item \emph{$\B$-Equivalence}: A $\B$-equivalence (or just equivalence) claim $\psi_1 \equiv \psi_2$ or $\psi_1 \not\equiv \psi_2$ states that $\fn(\psi_1)=\fn(\psi_2)$,  or $\fn(\psi_1) \neq \fn(\psi_2)$, respectively.
\end{itemize}
Verifier can pose challenges to Prover of these two types:
\begin{itemize}
  \item A \emph{partial evaluation challenge} \texttt{Challenge($\Pev\sigma
        \psi$)}  asks Prover to provide the result of evaluating the
        arithmetization $\Eval{\psi}$ on an assignment $\sigma \colon X
        \rightarrow \F$ that leaves at most one variable unassigned. If all
        variables are assigned, Prover answers with a field element. If one
        variable is unassigned, then Prover answers with a polynomial in one
        variable of degree at most two.
  \item A \emph{distinct asignment challenge} \texttt{ChallengeDistinct($\psi_1,
        \psi_2$)} asks Prover to provide an assignment to all variables such
        that $\Pev\sigma \Eval{\psi_1} \neq \Pev\sigma \Eval{\psi_2}$.
\end{itemize}

\paragraph{Claim normalization.} Once $\Claimset$ is initialized with Prover's
initial claims, Prover and Verifier execute \texttt{Normalize}($\Claimset$), a
procedure that replaces all claims of $\Claimset$, of all four types, by
$\F$-evaluation claims, while respecting CEHP. In particular, normalization
deals with $\B$-equivalence claims---which are very numerous in
circuits derived from model-checking executions---more efficiently than
\cite{CAV23}.  Here we only sketch the normalization procedure. 
A detailed description, including pseudocode and the proof that
CEHP is respected, can be found in Appendix \ref{subsc:normalizeapx}. 

A $\B$-evaluation claim $\Pev\sigma \fn(\varphi) = b$ is replaced by $\Pev\sigma
  \Eval{\varphi} = b$. A count claim $\Sigma \fn(\varphi) = k$ is replaced by
$\Pev\sigma \Eval{\varphi} = k \cdot 2^{-(\Abs{\Free(\varphi)})}$, where $\forall
  {x_i \in \Free(\varphi)}.~ \sigma(x_i) = 2^{-1}$ (the proof that this
replacement preserves CEHP is non-trivial;  it uses Lemma 2 of
\cite{CAV23}). For a $\B$-equality claim about $(\psi_1 \equiv \psi_2) =
  1$, Verifier samples an assignment $\sigma \colon \F \to X$ uniformly at random ($\R$),
sends the challenges $k_1 = \texttt{Challenge($\Pev\sigma \psi_1$)}$ and $k_2 =
  \texttt{Challenge($\Pev\sigma \psi_2$)}$ to Prover, and adds the two claims
returned by Prover to $\Claimset$ after checking $k_1 = k_2$. For a
$\B$-equivalence claim $(\psi_1 \equiv \psi_2) = 0$, Verifier sends the challenge
$\sigma = \texttt{ChallengeDistinct($\psi_1, \psi_2$)}$ to Prover, and adds the
claims $\Pev\sigma \Eval{\psi_1} = k_1, \Pev\sigma \Eval{\psi_2} = k_2$ returned
by Prover ($\sigma$ is now chosen by Prover) to $\Claimset$ after checking $k_1
  \neq k_2$. (Loosely speaking, CEHP is preserved by the DeMillo-Lipton-Schwartz-Zippel lemma: since $\Eval{\psi_1}$
and $\Eval{\psi_2}$ have total degree at most $n$, the polynomial $\Eval{\psi_1}-
  \Eval{\psi_2}$ has at most $n$ zeroes; so the polynomials differ almost
everywhere, and the probability that they differ for an assignment picked
uniformly at random is at least $1-n/|\F|$.)

\subsection{Detailed view of $\prot$: Rounds}
\label{subsubsec:round}
We describe the round of $\prot$ for a gate $\psi$. It consists of two parts: a
procedure \texttt{Merge}, that reduces all claims about $\psi$ to a unique
\emph{principal claim} while preserving CEHP, followed by a procedure
\texttt{Propagate} that replaces the principal claim about $\psi$ by one claim
for each successor of $\psi$\footnote{So, if a gate $\psi'$ has multiple
  predecessors $\psi_1, \ldots, \psi_k$, $\Claimset$ may contain up to $k$ claims
  about $\psi$.}. Again, pseudocode and  proofs can be found in Appendix
\ref{subsc:roundapx}.

\texttt{Merge} is taken from {\CPCertify} and was already described in
\cite{CAV23}; since its role is more technical, it is only presented in.
Here it suffices to note that \procname{Merge} preserves CEHP with probability
at least $1 - 2n/|\F|$. \texttt{Propagate} is the core of \prot. It consists of
four procedures, depending on whether the gate is labeled with a boolean
operation, a degree reduction, a projection, or a renaming.  We describe the
four cases, assuming that the principal claim is $\Pev\sigma\Eval{\psi} = k$ for
some $\sigma \colon X \to \F$ and some $k \in \F$. The first three are as in
\cite{CAV23}, and the last one is novel.

\medskip\noindent {\texttt{Propagate} [Binary Operation $\psi = \psi_1
      \circledast \psi_2$]:} Verifier sends Prover the challenges
\texttt{Challenge($\Pev\sigma \Eval{\psi_1}$)} and \texttt{Challenge($\Pev\sigma
    \Eval{\psi_2}$)}. Prover answers with claims $\Pev\sigma \Eval{\psi_1}=k_1$ and
$\Pev\sigma \Eval{\psi_1}=k_2$. Verifier checks the consistency condition,  $k_1
  \widehat{\circledast} k_2 = k$, where $\widehat{\circledast}$ is the operation
on polynomials satisfying $\Eval{p_1 \circledast p_2} = p_1
  \widehat{\circledast} p_2$; for example, $p_1 \widehat{\land} p_2 = p_1 \cdot
  p_2$ and $p_1 \widehat{\lor} p_2 = p_1 + p_2 - p_1p_2$. If the condition does
not hold, Verifier rejects Prover's initial claims, and otherwise replaces
$\Pev\sigma\Eval{\psi} = k$ by $\{ \Pev\sigma \Eval{\psi_1}=k_1, \Pev\sigma
  \Eval{\psi_1}=k_2\}$. This update trivially satisfies CEHP with probability 1.

\medskip\noindent {\texttt{Propagate} [Degree Reduction $\psi = \delta_{x_i} \psi'$]:}
Observe that, by the definition of degree reduction, $\Eval{\psi} = x_i \cdot
  \Eval{\psi'}|_{x_i:=1} + (1 -x_i) \Eval{\psi'}|_{x_i:=0}$. Verifier sends Prover
the challenge \texttt{Challenge($\Pev{\sigma'} \Eval{\psi'}$)}, where $\sigma'
  \colon X \setminus \{x_i\} \to \F$ is the assignment satisfying $\sigma'(y) =
  \sigma(y)$ for every $y \neq x_i$. Prover answers with a claim $\Pev\sigma
  \Eval{\psi'}=p(x_i)$, where $p(x_i)$  is a univariate polynomial of degree two.
Verifier checks the consistency constraint $\sigma(x_i) \cdot p(1) + (1-\sigma(x_i))
  \cdot p(0) = k$. If the condition does not hold, Verifier rejects. Otherwise,
Verifier picks $r \in \F$ u.a.r. and replaces  $\Pev\sigma\Eval{\psi} = k$ by
$\Pev{\sigma[x_i:= r]} \Eval{\psi'}=p(r)$. This update is shown to preserve CEHP
with probability at least $1 - 2/|\F|$ in \cite{CAV23}. Intuitively, if
$\Pev\sigma\Eval{\psi} \neq k$ then $\Eval{\psi'}$ and $p$ are different
polynomials of grade two, and so they differ almost everywhere (in at least $\F
  - 2$ points). Therefore, with probability at least $1- 2 / |F|$ we have
$\Pev{\sigma[x_i:= r]} \Eval{\psi'} \neq p(r)$.

\medskip\noindent {\texttt{Propagate} [Projection $\psi = \pi_{[i := b]} \psi'$]:}
Verifier replaces the claim $\Pev\sigma\Eval{\psi} = k$ by $\{\Pev{\sigma'}
  \Eval{\psi'} = k\}$, where $\sigma' := \sigma[x_i \rightarrow \Eval{b}]$. This update
trivially preserves CEHP.

\medskip\noindent{\texttt{Propagate} [Renaming $\psi=\psi'[x_l/x_t]]$:} Verifier
replaces the claim $\Pev\sigma\Eval{\psi} = k$ by $\Pev{\sigma'} \Eval{\psi'} =
  k$, where $\sigma' := \sigma[x_l :=\sigma(x_t)]$. This update preserves CEHP
only because $x_l$ does not occur in $\psi'$. (See Appendix
\ref{subsc:roundapx} for the proof.)
\subsection{Correctness}
%
The soundness and completeness of \prot\ follow from the CEHP preservation
properties of each of the subprocedures it uses:

\begin{lemma} If $\mathcal{C}$
  contains a false claim about an \gbcshort\ $\conv(\varphi)$ with $n$ variables,
  then Verifier accepts with probability at most $\left({4n |\varphi| +
      n}\right)/{\F}$ for any Prover. If all claims in $\mathcal{C}$
  are true, Verifier accepts with probability $1$ for the honest prover
  (completeness).
  \label{lemma:ebcertifycorrectness}
\end{lemma}

\begin{corollary}
  Any trace of \checker\ containing $N$ operations and using $n$ distinct
  boolean variables is rejected by Verifier if at least one assertion in the trace is wrong
  with probability at least $1 - \left({4n N + n}\right)/{\F}$, and
  otherwise accepted with probability $1$.
\end{corollary}
\begin{proof}
Traces of length $N$ result in GBCs of size $\leq N$.
\end{proof}

\section{BDD-Based Implementation of Prover}
\label{sc:prover}
\begin{figure}[t]
  \begin{center}
    {\footnotesize
      \begin{minipage}[t]{0.4\textwidth}
      \begin{algorithm}[H]
       \caption{{\Apply}($v_1, v_2, \circledast$)}
        \begin{algorithmic}
          \IF{$v_1 \in \{0, 1\}$ \OR $v_2 \in \{0, 1\}$}
          \RETURN $v_1 \circledast v_2$
          \ENDIF
          \STATE $\langle f_i,f_l,f_r \rangle = v_1$
          \STATE $\langle g_i,g_l,g_r \rangle = v_2$
          \IF{$f_i < g_i$}
          \STATE $f_l \leftarrow v_1$
          \STATE $f_r \leftarrow v_1$
          \ELSIF{$g_i < f_i$}
          \STATE $g_l \leftarrow v_1$
          \STATE $g_r \leftarrow v_1$
          \ENDIF
          \STATE $l \leftarrow {\Apply}(f_l, g_l, \circledast)$
          \STATE $r \leftarrow {\Apply}(f_r, g_r, \circledast)$
          \RETURN {\Reduce}$(\langle v_i, l, r \rangle)$
        \end{algorithmic}
        \end{algorithm}
      \end{minipage}
      \qquad
      \begin{minipage}[t]{0.50\textwidth}
        \newcommand{\pho}{}
        \begin{algorithm}[H]
       \caption{{\ApplyEBDD}($v_1, v_2, \circledast$)}
        \begin{algorithmic}
          \IF {$v_1 \in \{0, 1\}$ \OR $v_2 \in \{0, 1\}$}
          \RETURN $v_1 \circledast v_2$
          \ENDIF
          \STATE $\langle f_i, f_l, f_r \rangle = v_1$
          \STATE $\langle g_i, g_l, g_r \rangle = v_2$
          \IF {$f_i < g_i$}
          \STATE $f_l \leftarrow v_1$
          \STATE $f_r \leftarrow v_1$
          \ELSIF {$v_w < v_u$}
          \STATE $g_l \leftarrow v_2$
          \STATE $g_r \leftarrow v_2$
          \ENDIF
          \STATE $l \leftarrow$ {\ApplyEBDD}($f_l, g_l, \circledast$)
          \STATE $r \leftarrow$ {\ApplyEBDD}($f_r, g_r, \circledast$)
          \STATE $\mathit{final} \leftarrow$ \texttt{Reduce}($\langle v_i, \link(\link(l)), \link(\link(r)) \rangle$)
          \STATE $\mathit{node} \leftarrow \langle v_i, l, r \rangle$; $\link(\mathit{node}) \leftarrow \mathit{final}$
          \STATE $b \leftarrow \langle v_1 \circledast v_2 \rangle$; $\link(b) \leftarrow \mathit{node}$
          \RETURN $b$
        \end{algorithmic}
        \end{algorithm}
      \end{minipage}
    }
  \end{center}
  \caption{Comparison of {\Apply}($v_1, v_2, \circledast$) and our new
  algorithm {\ApplyEBDD}($v_1, v_2, \circledast$).}
  \label{fig:applyandnew}
\end{figure}

$\prot$ describes \emph{what} Prover has to do (compute certain polynomials and
evaluate them), but not \emph{how} to do it. We give an implementation of Prover
that improves on the one presented in \cite{CAV23}. In
Section~\ref{subsec:CAV23Prover}, we briefly recall the implementation of
\cite{CAV23}. In Section~\ref{subsec:shortcomings} we explain its shortcomings
and sketch our novel implementation.

\subsection{Implementation of Prover by Couillard \etal \cite{CAV23}}
\label{subsec:CAV23Prover}
At first sight, the computations performed by Solver and Prover on an instance
of the model-checking problem seem to be unrelated: Solver computes BDDs for the
(boolean functions of) the gates of a circuit $\varphi$, while Prover computes
polynomials for the gates of $\conv(\varphi)$. More precisely, consider a gate
$\psi = \psi_1 \circledast \psi_2$ of a circuit $\varphi$.  After
degree-reductions with the gate $\psi' := \delta_{x_n} \delta_{x_{n-1}} \cdots
\delta_{x_1} (\psi_1  \circledast \psi_2)$ of $\conv(\varphi)$ corresponding to
$\varphi$, then:
\begin{itemize}
  \item Solver's task is to compute a BDD-node for $\fn(\psi)$ from BDD-nodes
  for $\fn(\psi_1)$ and $\fn(\psi_2)$.  For this, Solver uses a well-known
  recursive algorithm {\Apply}($v_1, v_2, \circledast$), where $v_1$ and
  $v_2$ are the unique BDD nodes representing $\fn(\psi_1) , \fn(\psi_2)$.
  \texttt{Apply} is shown on the left of Figure \ref{fig:applyandnew}.
  \item Prover's task  is to compute polynomials for each of the gates
        $g_0:=\psi_1 \circledast \psi_2$, $g_1:= \delta_{x_{n-1}} g_0$, $g_2:=
        \delta_{x_{n-2}} g_1$, \ldots, $g_n=\delta_{x_1} g_{n}$, where $x_1,
        \ldots x_n$ are the free variables of $\psi_1 \circledast \psi_2$, and evaluate
        them at assignments chosen by Verifier.
\end{itemize}
Couillard \etal show in \cite{CAV23} that, \emph{if polynomials are encoded
using an appropriate data structure}, then Prover does not have to compute them
because, surprisingly,  they are already computed by Solver. More precisely,
Couillard \etal transform \texttt{Apply($v_1, v_2, \circledast$)} into
\texttt{ComputeEBDD($v_1, v_2, \circledast$)}, another algorithm which, despite
having the same runtime as {\Apply}($v_1, v_2, \circledast$), computes not
only the BDD-node for $\fn(\psi)$ but also encodings for all the polynomials in
the data structure. The Solver of \cite{CAV23} just runs
\texttt{ComputeEBDD($v_1, v_2, \circledast$)} instead of \texttt{Apply($v_1,
v_2, \circledast$)}. Appendix~\ref{apx:newname} describes the algorithm in
detail.
 
The data structure is called \emph{extended BDDs} (eBDDs). Formally, an eBDD
node $e$ is either $0$, $1$, a node $\langle x, e_0, e_1 \rangle$, where $x \in
X$ and $e_0, e_1$ are eBDD nodes, called the 0-child and 1-child of $e$,
or---and this is the extension---a \emph{binary operation node} $\langle v_1
\circledast v_2 \rangle$, where $\circledast$ is a binary boolean operator and
$v_1, v_2$ are BDD nodes. The semantics of an eBDD, say $e$, is the polynomial
$\Eval{e}$ defined by

\smallskip\begin{centerline}{\small $\Eval{0}:=0$ \hspace{0.1cm} $\Eval{1}:=1$ \hspace{0.1cm}
$\Eval{\langle x, e_0, e_1 \rangle} := x \cdot \Eval{e_1} + (1-x) \cdot
\Eval{e_0}$ \hspace{0.1cm}  $\Eval{\langle v_1 \circledast v_2 \rangle} := \Eval{v_1}
\widehat{\circledast} \Eval{v_2}$}
\end{centerline}

\smallskip \noindent  Figure \ref{fig:ebddcomp}
  shows eBDDs encoding the polynomials $\Eval{\psi_1 \vee \psi_2}$,
  $\delta_{x_1} \Eval{\psi_1 \vee \psi_2}$, and $\delta_{x_2} \delta_{x_1}
  \Eval{\psi_1 \vee \psi_2}$ for  $\psi_1 = x_1$ and $\psi_2 =  x_1 \wedge x_2$. Binary
  operation nodes are shaded blue.
  
\smallskip \noindent Observe that evaluating a polynomial on an assignment takes
linear time in the size of the eBDD encoding it. For example, in order to
compute $\Eval{\langle x, e_0, e_1 \rangle}(\sigma)$ for an assignment $\sigma$
we just use $\Eval{\langle x, e_0, e_1 \rangle}(\sigma) = \sigma(x) \cdot
\Eval{e_1}(\sigma) + (1 -\sigma(x)) \cdot \Eval{e_0}(\sigma)$.

\begin{figure}[t]
  \centering
  \includegraphics[width=\textwidth]{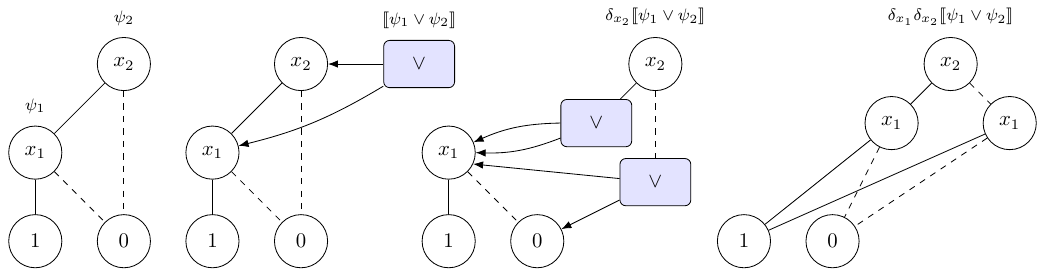} 
  \caption{eBDDs for $\Eval{\psi_1 \vee \psi_2}$, $\delta_{x_1} \Eval{\psi_1 \vee \psi_2}$, and
  $\delta_{x_2} \delta_{x_1} \Eval{\psi_1 \vee \psi_2}$, where $\psi_1 = x_1$ and $\psi_2 =  x_1 \wedge x_2$, and BDD obtained after simplifying the latter.}
  \label{fig:ebddcomp}
\end{figure}


\subsection{Improving {\ComputeEBDD} and {\blic}}
\label{subsec:shortcomings}
While {\ComputeEBDD}($v_1, v_2, \circledast$) has the same runtime as
{\Apply}($v_1, v_2, \circledast$), it has two strong shortcomings in practice:
\begin{itemize}
  \item {\Apply}($v_1, v_2, \circledast$) is a recursive algorithm with
  memoization. Like all recursive algorithms, it produces a tree of recursive
  calls that are processed in depth-first manner using the recursion stack.
  Loosely speaking, {\ComputeEBDD} explores the same tree, but in
  breadth-first manner, which prevents the re-use of the recursion stack across
  multiple invocations.
  \item BDD-libraries are implemented on top with a common optimization: A
  \emph{global computation cache} containing BDD-nodes for \emph{all} of the boolean
  (sub-\nolinebreak)functions computed so far.  If two functions use the
  same sub-function as part of their logic, they will share the computation of its BDD.
  However, {\ComputeEBDD} is incompatible with a global cache. In order to
  match the runtime of {\Apply}, it uses a mutable data structure that
  stores eBDDs for the polynomials in-place, instead of creating new nodes. The
  mutations are written to an \emph{undo-log}, which can be applied by Prover to
  access overwritten eBDDs. While a global cache was not necessary for \blic,
  using one tends to highly benefit model checking, with \cite{Yang1998} finding
  a reduction of repeated sub-computations of at least an order of magnitude.
\end{itemize}

\begin{figure}[ht]
  \centering
  \includegraphics[width=0.7\textwidth]{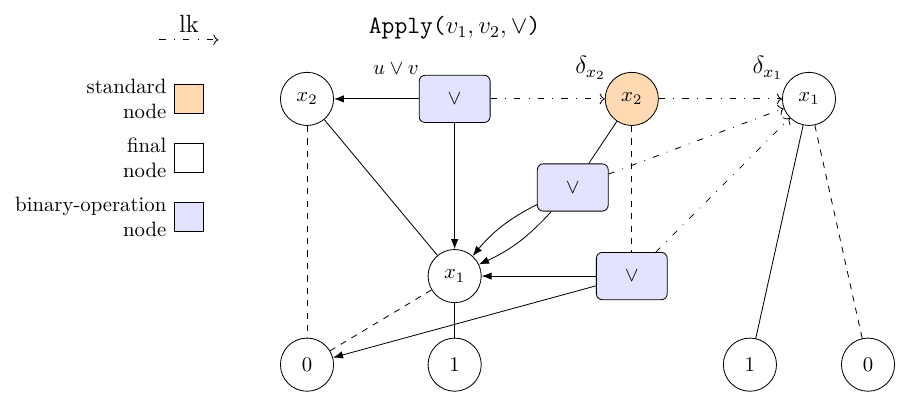}
  \caption{New computation of immutable eBDDs using links for versioning.}
  \label{fig:ebddcomputenew}
\end{figure}

We present a new algorithm {\ApplyEBDD},  shown on the right of Figure
\ref{fig:applyandnew}, that solves these two problems. {\ApplyEBDD} runs on
top of a global cache. On top of the links to its children, every eBDD node $w$
has an additional  \emph{link}, accessible via $\link(w)$. Intuitively
$\link(w)$ is used to access multiple versions of the \emph{same} eBDD node as a
linked-list. This data structure allows for an arbitrary number of versions for
each node, but in {\ApplyEBDD} it is always limited to at most three: one
for the initial, binary operation node; a second version for an eBDD node that
has some binary operation node as a descendant (\emph{standard} eBDD node), and
a third fully reduced \emph{final} BDD node. Figure~\ref{fig:ebddcomputenew}
shows the resulting immutable eBDDs that encode each degree reduction
({\ComputeEBDD} would produce the eBDDs shown in
Fig.~\ref{fig:ebddcompute}). Note that any duplicate final nodes are only
pictured for easier illustration and refer to the same object in the unique
table. Appendix~\ref{apx:newname} contains more
details on {\ApplyEBDD}.

The only difference between  {\ApplyEBDD}  and {\Apply} is that after
computing the two recursive results $l$ and $r$, it creates \emph{three} new
eBDD nodes instead of one, and connects them using the link field. Prover
computes the eBDD for any polynomial of the form $\delta_{x_k} \dotsb
\delta_{x_1} \Eval{u} \circledast \Eval{v}$ for all $1 \leq k \leq n$ by
evaluating BDD-nodes $d$ with $\var(d) > x_k$ as $\link(d)$, and otherwise
ignoring links. The unique final BDD is available by evaluating the root as
$\link(\link(r))$.

Since {\ApplyEBDD} follows the depth-first recursive structure of
{\Apply}, and does not mutate any existing nodes, the values it returns can easily
be tabled according to their recursion parameters. The table can
either be global, or cleared at any point throughout program execution.
We prove in Appendix \ref{apx:proposition1}:

\begin{proposition}\label{prop:computebdd} Let $\psi_1,\psi_2$ denote nodes of
  $\conv(\varphi)$ and $u_1,u_2$ BDDs with $\Eval{u_i}=\Eval{\psi_i}$,
  $i\in\{1,2\}$. Then {\ApplyEBDD}$(u_1, u_2, \circledast)$ satisfies
  $\Eval{w_0}=\Eval{\psi_1\circledast\psi_2}$ and
  $\Eval{w_{i+1}}=\delta_{x_{n-i}}\Eval{w_{i}}$ for every $0 \leq i \leq n-1$;
  moreover, $w_n$ is a BDD with $w_n={\Apply}(u_1,u_2, \circledast)$.
  Finally, the algorithm runs in time $O(T)$, where $T$ is the time taken by
  {\Apply}($u_1, u_2, \circledast$).
\end{proposition}

Thus, Solver and Prover use {\ApplyEBDD} to compute eBDDs for each boolean
function and polynomial of the {\gbcshort} constructed by \lib. Prover answers a
\procname{Challenge} and \procname{ChallengeDistinct} by traversing the
arguments eBDDs in linear time. Pseudocode and a detailed description is
provided in Appendix~\ref{apx:challengeprover}.


\subsection{Garbage collection}
\label{subsec:gc}

Recall that {\checker}'s Solver runs bottom-up through a circuit: For every gate
$\psi := \psi_1  \circledast \psi_2$ of $\conv(\varphi)$, Solver applies
{\ApplyEBDD} to $\psi$ (that is, computes {\ApplyEBDD}($\psi_1, \psi_2,
\circledast$)), \emph{after} applying it to $\psi_1$ and $\psi_2$. In
particular, Solver can garbage-collect all BDD-nodes of $\psi_1$ or $\psi_2$
that are not shared with $\psi$, reducing memory consumption. On the contrary,
{\prot} runs top-down: the round for $\psi$ is executed \emph{before} the rounds
for $\psi_1$ and $\psi_2$. Therefore, since Prover uses the output Solver for
all gates of the circuit, garbage collection is not possible. 

We show that, under a reasonable assumption, {\prot} can be replaced by a
{\protGC}, another protocol that essentially runs {\prot}
bottom-up (\texttt{Rev} stands for ``reverse'').  We explain the intuition for {\protGC} with the help of the tiny
circuit $\varphi = x \wedge x$, for which $\conv(\varphi) =  \delta_x(x \wedge
x)$. Imagine a dishonest Prover claims $\Eval{\conv(\varphi)} (1/2) = 1$, which
corresponds to claiming that $x \wedge x$ has two satisfying assignments.  Then
$\prot$ runs as follows:
\begin{itemize}
\item[(1)] Verifier asks Prover to supply $\Eval{x \wedge x}$. \\
Let $p(x)$ be Prover's answer. Verifier checks that $p(x)$ is at most quadratic
and the consistency condition $(x \cdot p(1) + (1-x) \cdot p(0)) (1/2) =
1/2(p(1)+p(0)) \stackrel{?}{=}1$, and rejects if it they are not met. In
particular, if Prover answers the truth, namely $p(x):=x^2$, then Prover is
caught. So Prover answers with some $p(r) \neq r^2$.
\item[(2)] Verifier picks $r \in \F$ u.a.r. and asks Prover to supply $\Eval{x}(r)$. \\
Let $k$ be Prover's answer. Verifier checks the condition  $p(r) \stackrel{?}{=}
k^2$, and rejects if it is not met. Assume it holds.
Then Verifier computes $\Eval{x}(r)$ herself and checks $\Eval{x}(r)=r
\stackrel{?}{=}k$, which is equivalent to $p(r)\stackrel{?}{=}r^2$. But for
$p(r) \neq r^2$ this holds only if $r$ happens to be one of the at most two
roots of the quadratic polynomial $p(r) - r^2$, and so with probability
$2/|\F|$. So Verifier catches that Prover's initial claim is false with high
probability.
\end{itemize}
It is essential that the procedure runs top down, i.e., that (2) happens after
(1). Otherwise Prover knows $r$ when choosing $p(x)$, and Prover can choose
$p(x) := \frac{(r-2)x^2+rx}{r-1}$, passing all checks.

Imagine, however, that Prover behaves like an \emph{oracle}, i.e., that it
clears its memory after each query. Then in the bottom-up protocol in which (2)
happens first, Prover forgets $r$, and then (1) happens, Prover is caught with
the same probability as before. This condition is reasonable whenever Verifier
can assume that Solver+Prover may be faulty but are not malicious and do not
conspire to ``fool'' Verifier. For example, if the code of Prover is publicly
available, then Verifier can check, e.g.\ by program analysis, that Prover does
not store data across queries.

Appendix~\ref{apx:gccertify} describes {\protGC}, a bottom-up version of
{\prot} with garbage collection. Here we only sketch it. Verifier starts {\protGC} by creating a random assignment $\sigma$ for all variables.
Then {\protGC} proceeds in rounds, one for each gate, in \emph{reverse}
topological order, i.e., from input to output gates. At the round for a gate
$\psi$, Verifier replaces claims about the output gates of $\psi$ by a claim
about $\psi$ itself. CEHP is preserved by the oracle assumption on Prover.
At the end of {\protGC}, Verifier accepts if at the end $\Claimset$ contains claims for all  assertions of the execution trace.
In Appendix~\ref{apx:gccertify} we prove:

\begin{lemma} If $\mathcal{C}$ contains a
  false claim about an {\gbcshort} $\conv(\varphi)$ with $n$ variables, then in {\protGC} Verifier
  accepts with probability at most $\left({4n |\varphi| + n}\right)/{\F}$ for
  any Prover that acts as an oracle. If all claims in $\mathcal{C}$ are true,
  Verifier accepts with probability $1$ for the honest prover (completeness).
  \label{lemma:eccertifycorrectness}
\end{lemma}

\section{Evaluation}
\label{sec:eval}

We evaluate {\checker} on benchmarks from the liveness track of
HWMCC25's~\cite{Preiner2025}. We convert Aiger benchmarks
to the \textsf{smv} format using \textsf{aigtosmv}~\cite{Biere2025}, and retain the 53 benchmarks for which \checker\ terminates within $15$ minutes. All experiments are executed
on a platform running Linux 6.15.2 with an AMD `Ryzen 9 7950X' CPU, 32GB of DDR5
memory, and hosting the Clang compiler in version 19.1.7.

\smallskip\noindent\textit{Performance of {\checker}'s Solver.} Fig.~\ref{fig:solvervsolver} compares the runtime and memory consumption of {\checker}'s Solver and NuSMV 2.7.0 which, recall, does not offer certification. Our solver is in average \nusmvoverhead\ times slower, which for a  first prototype against a mature tool we consider a good result. The main reason is that {\checker}'s Solver does not yet support
\emph{relational products}~\cite{Burch1994}---a BDD operation that can
significantly improve the performance of symbolic model checking---because its direct interactive certification by Prover and
Verifier is still an open problem.  

\begin{figure}[ht]
  \centering
  \begin{minipage}{0.495\textwidth}
    \centering
    \includegraphics[width=\textwidth]{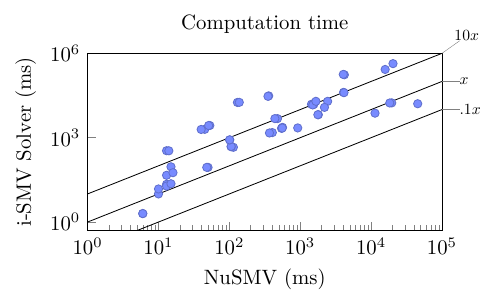}
  \end{minipage}
  \begin{minipage}{0.495\textwidth}
    \centering
    \includegraphics[width=\textwidth]{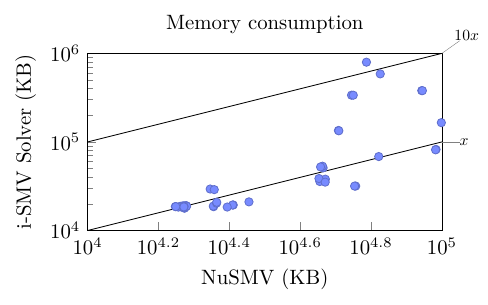}
  \end{minipage}
  \begin{minipage}{0.495\textwidth}
    \centering
    \includegraphics[width=\textwidth]{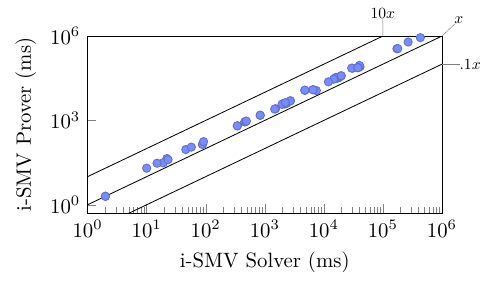}
  \end{minipage}
  \begin{minipage}{0.495\textwidth}
    \centering
    \includegraphics[width=\textwidth]{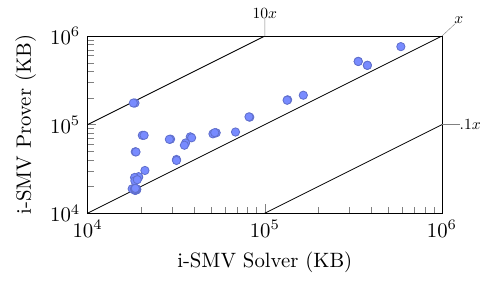}
  \end{minipage}
  \caption{Top: Solving time (left) and memory consumption (right) of solving instances using \checker\ vs.\ NuSMV.
  Bottom: {\checker}'s Solver vs. {\checker}'s Prover computation time.}
  \label{fig:solvervsolver}
\end{figure}

\smallskip\noindent\textit{Performance of {\checker}'s Prover.} Fig.~\ref{fig:solvervsolver}
(bottom row) compares the time and memory consumption of \checker's Solver and
\checker's Prover. Our theoretical analysis shows that if Solver takes time $T$,
then Prover takes time $\O(T)$, and suggests a small constant,
as answering all challenges amounts to traversing the BDDs computed by Solver
once. The experiments confirm this, the constant being \proveroverhead\ on average. Note that Prover time includes both computing all eBDDs (solving the instance) and answering all challenges
(certifying the result).

\smallskip\noindent\textit{Performance of {\checker}'s Verifier.} 
Fig.~\ref{fig:provervverifier} (top row) compares the runtime of {\checker}'s
Verifier and {\checker}'s Prover.  For a timeout
of 15 minutes, Verifier never takes more than \verifiermaxtime\ second and on
average Verifier is \verifierspeedup\ times faster than Prover. Further,
Verifier's runtime grows much slower than Prover's runtime
with the size of the instance in average. This reflects the fact that Verifier's runtime is
linear on the \emph{length} of the execution trace, while Prover's time is
linear on the \emph{time} it takes to execute it. Generally,  execution traces
for larger benchmarks manipulate larger BDDs, which improves
the speedup of Verifier w.r.t. Prover; however, some have long trace sand small BDDs, and Verifier still
takes almost as much time as Prover.


\begin{figure}[ht]
  \centering
  \begin{minipage}{\textwidth}
    \centering
    \includegraphics[width=.5\textwidth]{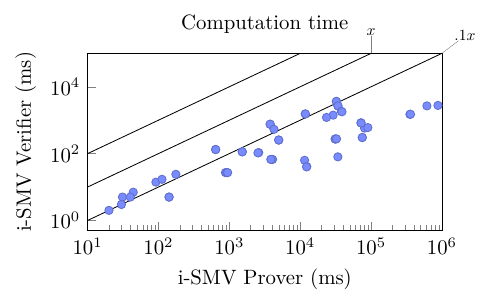}
  \end{minipage}
  \begin{minipage}{0.495\textwidth}
    \centering
    \includegraphics[width=\textwidth]{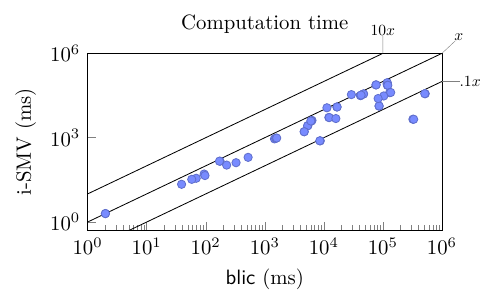}
  \end{minipage}
  \begin{minipage}{0.495\textwidth}
    \centering
    \includegraphics[width=\textwidth]{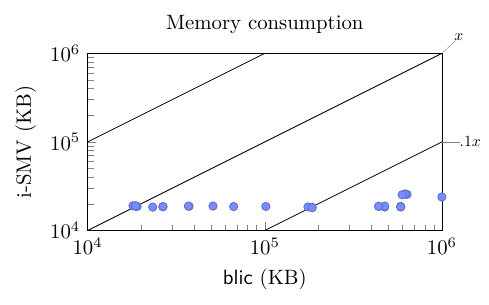}
  \end{minipage}
  \caption{Top: Computation time of \checker's Prover vs.\ Verifier.
    Bottom: \checker\ using our improvements vs.\ \blic.}
  \label{fig:provervverifier}
\end{figure}


\noindent\textit{Impact of {\prot} and {\blictwo}.}
Fig.~\ref{fig:provervverifier} (bottom row) compares the performance of
{\checker} when run on top of {\CPCertify}+{\blic} and and on top of
{\protGC}+{\blictwo} (that is, with dedicated handling of renaming gates and
equivalence claims, with a BDD library that uses a global unique table and with
fine-grained garbage collection). The speedup factor is \improvementsspeedup\ on
average, increases with runtime, and reaches a maximum of 73. The  memory-reduction factor is
\improvementmemory\ on average, with median of \improvementmemorymedian\ for the benchmarks that
did not run out of memory ({\CPCertify}+{\blic} solves
\naivetimeout\ fewer instances than {\protGC}+{\blictwo}).

\section{Conclusion}
We have presented  {\checker}, the first model checker with interactive certification. {\checker}'s Verifier module interactively checks that the execution sequence of the Solver module---that is, the sequence of BDD-operations executed to solve a given model-checking instance---is correct. 

Our certification technique works for any algorithm implemented on top of the {\lib} BDD-library. In particular, 
since there exists a BDD-based algorithm for the full modal $\mu$-calculus, it can be used to construct a self-certifying model checker not only for full CLT, but for the full modal $\mu$-calculus. It can also be used to certify a CTL model-checker that computes predecessors by iterative squaring of the transition relation. The execution sequence of this algorithm always has polynomial length in the size of the model-checking instance, and so for this algorithm Verifier always runs in polynomial time. However, this algorithm is known to be much less efficient for Solver.



\bibliographystyle{plain}
\bibliography{blic.bib}


\appendix

\counterwithin{figure}{section}
\counterwithin{definition}{section}
\counterwithin{lemma}{section}
\counterwithin{theorem}{section}
\counterwithin{corollary}{section}
\counterwithin{algorithm}{section}

\section{Preliminaries}

\subsection{Binary Decision Diagrams (BDDs)}
\label{apx:bdds}
We recall basic notions on reduced ordered binary decision diagrams~\cite{Bryant1986,Bryant2018}, called in this paper BDDs for short. BDDs are a symbolic representation of boolean functions as directed acyclic
graphs (DAGs). The graph forms a decision diagram: Each node of the graph,
starting at the root, is associated with a variable of the function. The node
has two children: one `low' child for the boolean function when setting the
node's variable to false, and one `high' child for the function of setting its
variable to true. We write $\langle x_i, l, r \rangle$ for a node with variable
$v_i$, the low child $l$, and the high child $r$. The leaf nodes are either $0$
or $1$. To evaluate the function represented by a BDD, one chooses the path from
its root to a leaf according to the value of each variable. If the last node is
$0$, then the function evaluates to false, otherwise it evaluates to true. BDDs
are \emph{ordered}, because any sequence of variables encountered on a path of a
BDD are sorted according to a strict global order. Further, BDDs are
\emph{reduced} since they cannot have redundant nodes: Whenever both children of
a node represent the same function, the node should be omitted entirely and
replaced by its child. Also, reduced BDDs never have a duplicate node (e.g.\ by
using a cache or table). This ensures that any boolean function has a
\emph{unique} representation as a BDD~\cite{Bryant1986}. 

\begin{figure}[ht]
  \centering
  \includegraphics[width=\textwidth]{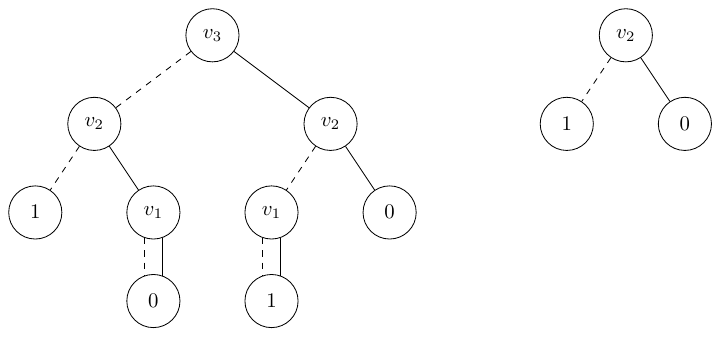}
  \caption{Unreduced (left) and reduced (right) ordered BDD representing the same function.}
\end{figure}

Boolean operations on binary decision diagrams are implemented via recursion on
the children of the operands. A BDD $\langle x_i, l, r \rangle$ represents a
boolean function $f$ according to the Shannon expansion $f = (\neg x_i \wedge
  f|_{x_i \leftarrow 0}) \vee (x_i \wedge f|_{x_i \leftarrow 1})$ where $l =
  f|_{x_i \leftarrow 0}$ and $r = f|_{x_i \leftarrow 1}$. Boolean operations can
be expressed using the subterms of the Shannon expansion, for example as: $f
  \wedge g = (\neg x_i \wedge f|_{x_i \leftarrow 0} \wedge g|_{x_i \leftarrow 0})
  \vee (x_i \wedge f|_{x_i \leftarrow 1} \wedge g|_{x_i \leftarrow 1})$. This
naturally gives rise to the \texttt{Apply($f, g, \circledast$)} operation to
compute the reduced BDD for the function $f \circledast g$, where $\circledast$
is one of the $16$ possible boolean binary operators.

Where $\texttt{Reduce}$ ensures the resulting BDD is unique and reduced by using
a \emph{unique cache}. As is, the $\texttt{Apply}$ algorithm produces correct
BDDs. However, the worst-case complexity is exponential in the size of its
parameters. This can be avoided by caching results using a \emph{computation
cache}, which prevents \texttt{Apply($f, g, \circledast$)} being computed twice
for any equal parameters. Because the unique cache is needed to correctly
generate reduced ordered binary decision diagrams, entries cannot be removed
from it while the entries' node exists. The computation cache, on the other
hand, can be limited to a maximum size in a trade-off between speed and memory
overhead. The use of a computation cache limits the number of recursive calls to
\procname{Apply($f, g, \circledast$)} to $|f| \cdot |g|$, bounding the total
execution time of an operation to $O(|f| \cdot |g|)$.

Table~\ref{table:bddops} outlines the typical BDD-based boolean function library
interface. The \texttt{Restrict} and \texttt{Rename} operations are also defined
by simple recursion on the graph structure of a BDD, using a computation cache
for efficiency. The difference between the two operations and \texttt{Apply} is
that the terminal case differs: While \texttt{Apply} stops recursion at the leaf
nodes, they stop as soon as their recursion reaches a node with a variable lower
than their argument. In general, both operations have a complexity linear in the
size of their BDD argument. The renaming operation \texttt{Rename($g, x_i,
    x_k$)}, which substitutes the variable $x_i$ for $x_k$ in the BDD $g$, requires
that none of the variables in between $x_k$ and $x_i$ is in the support of $g$.
The condition ensures that each corresponding variable in the graph of $g$ can
be renamed without violating the BDD's order.

Quantifying a set of variables $\{x_0, \ldots, x_n\}$ using BDDs is implemented by a
sequence of quantifications $\exists_{x_0} \dots \exists_{x_n} f$, which may
have a time complexity of $O(|f|^{2^n})$. This operation in particular can quickly become
limiting for model checking as the algorithm quantifies a set of variables for
every image and pre-image computation.

\subsection{Interactive Proof Protocols}
\label{apx:ip}

This appendix is taken from \cite{Arora2009}, with slight differences in notation. We first introduce deterministic interactive proof protocols, and then the general notion of an interactive proof protocol.

A deterministic interactive proof protocol is a pair of deterministic Turing machines, called the \emph{honest Prover} and the \emph{Verifier}, that compute two functions $f, g : \{0, 1\}^* \rightarrow \{0, 1\}^*$. Intuitively, the \emph{honest Prover} and the \emph{Verifier} are the two parties of a communication protocol, and $f, g$ describe their behavior: Given the sequence of messages  exchanged between the two parties so far, modeled as a word $w \in \{0, 1\}^*$, the strings $f(w), g(w)$ model the next message sent by the honest Prover to Verifier resp. by Verifier to the honest Prover.  Formally, a $k$-round interaction between the honest Prover and Verifier on a word $x \in \{0,1\}^*$ is defined as follows:

\begin{definition}[$k$-round interaction]
For any two Turing machines computing functions $f, g : \{0, 1\}^* \rightarrow \{0, 1\}^*$ and an input $x \in \{0,1\}^*$, a $k$-round interaction $\langle f, g\rangle_k(x)$ is defined as the sequence of binary strings $a_1,\dots,a_k$:
  \begin{align*}
    a_1                                 & = f(x)                                                        \\
    a_{2i}                              & = g(x, a_1, \dots, a_{2i - 1}) &  & \forall i.~ 1 < 2i \leq k \\
    a_{2i+1}                            & = f(x, a_1, \dots, a_{2i})     &  & \forall i.~ 1 < 2i < k    \\
    \text{out}\langle f,g \rangle_k(x) & = f(x, a_{1},\dots,a_{k})
  \end{align*}
\end{definition}

As usual, we model computational decision problems, like SAT or the model-checking problem for CTL, as the language containing the encodings of the ``yes''-instances: for SAT, the set of all satisfiable boolean formulas, and for the model-checking problem the set of all pairs $\langle\text{system, CTL-property} \rangle$ such that the system satisfies the property. We now formally define a $k$-round deterministic interactive proof  protocol for a language. 

\begin{definition}[$k$-round deterministic interactive proof  protocol]
Let $L \subseteq \{0,1\}^*$ be a language and let $k \colon \mathbb{N} \to \mathbb{N}$. A {\em $k$-round interactive proof protocol} for $L$ is a pair $HP$, $V$ of deterministic Turing machines, called the \emph{honest Prover} and the \emph{Verifier}, satisfying the following properties for every input $x$: 
\begin{itemize}
\item \textbf{Polynomiality}: $V$ runs in polynomial time in $|x|$.
\item \textbf{Completeness}:   if $x \in L$, then $\text{out}\langle V, HP \rangle_{k(|x|)}(x) = 1$.       
\item \textbf{Soundness}:  if $x \notin L$, then $\text{out}\langle V, P \rangle_{k(|x|)}(x) = 0$ for every deterministic Turing machine $P$.
\end{itemize}
\end{definition}
Intuitively, completeness means that for every $x \in L$ the honest Prover makes Verifier accept the true claim ``$x$ belongs to $L$'' (Verifier outputs $1$). Soundness means that for every $x \notin L$, \emph{no Prover whatsoever}, honest or dishonest, can make Verifier accept the false claim ``$x$ belongs to $L$''. 

As shown in \cite{Arora2009}, the power of interaction is only realized when Verifier is \emph{probabilistic}. So we need to introduce probabilistic Turing machines.

\begin{definition} [Probabilistic Turing Machine]
  A probabilistic Turing machine $M$ is a Turing machine that has a
  second \emph{random} input tape $\mathbb{R}$ and a function $t$, such that for
  any input $x \in \{0,1\}^*$, $\mathbb{R}$ is initialized with a bitstring $r$, sampled uniformly
  at random from $\left\{0, 1\right\}^{t(|x|)}$. We say that $M$ takes time $T$
  if for every $x$, $M$ terminates in at most $T(|x|)$ steps for every 
  $r \in \left\{0, 1\right\}^{t(|x|)}$.
\end{definition}

It is easy to extend the definition of $k$-round interaction to the case in which Verifier is probabilistic. 
We add the bitstring $r$ as input the function $f$, that is, we take $a_1 = f (x, r)$,
$a3 = f(x, r, a_1, a_2)$, etc. The interaction $\langle V, HP \rangle(x)$ is now a random variable over $r \in \{0,1\}^{t(|x|)}$. Similarly the output $\text{out}_f\langle V, P \rangle(x)$ is also a random variable. Now we can generalize deterministic interactive proof protocols to interactive proof protocols:

\begin{definition}[$k$-round interactive proof  protocol]
Let $L \subseteq \{0,1\}^*$ be a language and let $k \colon \mathbb{N} \to \mathbb{N}$. A {\em $k$-round interactive proof protocol} for $L$ is a pair $HP$, $V$ of deterministic and probabilistic Turing machines, respectively, called the \emph{honest Prover} and the \emph{Verifier}, satisfying the following properties for every input $x$ : 
\begin{itemize}
\item \textbf{Polynomiality}: $V$ runs in polynomial time in $|x|$.
\item \textbf{Completeness}:   if $x \in L$, then $\Pr[\text{out}_f\langle V, P \rangle(x) = 1] = 1$.       
\item \textbf{Soundness}:  if $x \notin L$, then $\Pr[\text{out}_f\langle V, P \rangle(x) = 1] \leq \frac{1}{2^{|x|}}$ for every  deterministic Turing machine $P$.
\end{itemize}
\end{definition}

In other words: the honest Prover makes Verifier accept the true claim that $x$ belongs to $L$ with probability 1. Soundness means that for every $x \notin L$, \emph{no Prover whatsoever}, honest or dishonest, can make Verifier accept the false claim that $x$ belongs to $L$ with probability higher than $\frac{1}{2^{|x|}}$. 

Finally, we define the class  \class{IP} of decision problems as the problems for which there exists a
$k$-round interactive proof protocol for some number of rounds $k$ polynomial in the size of the input. 

\begin{definition}[\class{IP}{[$k$]}] For any polynomial $k(n) > 0$, a language $L$ is in $\class{IP}[k]$ if there is a $k$-round interactive proof protocol for $L$.  We define the class $\class{IP}$ of problems with interactive proof systems as $\class{IP} = \bigcup_{c > 0} \class{IP}[n^c]$. \label{def:IP}
\end{definition}

Shamir's theorem states $\class{IP} = \class{PSPACE}$. In other words, for every problem in $\class{PSPACE}$---like the model checking problem for CTL where the set of initial configurations, the transition function, and the atomic propositions are given as BDDs---has an interactive proof protocol.

\section{Preprocessing: From $\varphi$ to $\conv(\varphi)$}
\label{apx:preprocess}

As mentioned in the main text, \prot\ transforms a model checking trace to the
\emph{\gbclong} (\gbcshort) $\conv(\varphi)$. A \gbcshort\ is a directed,
acyclic graph in which each node represents a boolean function. For example,
fig.~\ref{fig:ebc} displays a \gbcshort\ together with the boolean functions encoded
by its nodes. Importantly, a \gbcshort\ can \emph{share} common sub-expressions, which
allows $\conv(\varphi)$ to be linear in the size of the trace.

\begin{figure}[t]
  \centering
  \includegraphics[width=.6\textwidth]{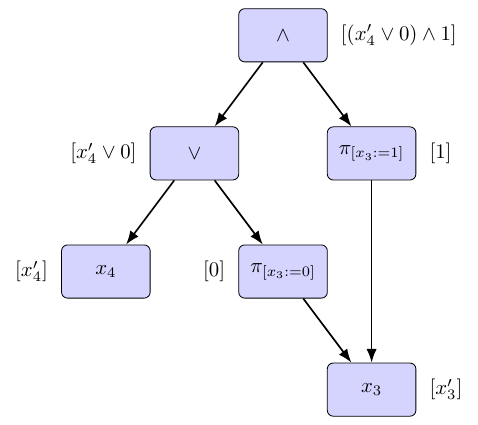}
  \caption{Extended boolean circuit and each node's corresponding boolean function.}
  \label{fig:ebc}
\end{figure}

In addition to a boolean function $\fn(\varphi)$, each node $\varphi$ is
associated with its \emph{arithmetization} $[\![\varphi]\!]$. As said in
Section~\ref{subsec:preproc}, $[\![\varphi]\!]$ is a polynomial over finite
field $\F$, which is equal to $\fn(\varphi)$ when each variable is
assigned a boolean value. The paper explained how boolean operations get mapped
to their arithmetization. To convert each node of a \gbcshort, we also define an
additional \emph{partial evaluation} operator for polynomials:

\begin{definition}[Polynomials]
  \begin{itemize}
    \item \emph{Partial evaluation ($\pi_{[x:=a]} p$)}: returns the polynomial of
          setting variable $x$ to $a \in \F$ in $p$.
    \item \emph{Degree reduction ($\delta_x p$)}: returns the polynomial
          obtained by setting the current degree $d$ of $x$ in every monomial of
          $p$ to $\max(d, 1)$.
  \end{itemize}
  Adjacent partial evaluations or degree reductions are commutative, e.g.\
  $\delta_x \delta_y p = \delta_y \delta_x p$. Also, if a variable $v$ does not
  occur in $p$, then $\delta_v p = p$ and $\pi_{[v:=a]} p = p$. We use these two
  facts throughout proofs without explicit mention.

  A partial evaluation of $p$ under a partial assignment $\sigma : X'
    \rightarrow \F, X' \subseteq V'$, written $\Pev\sigma p$, is defined as
  $\pi_{[x_1:=\sigma(x_1)]}\dots\pi_{[x_k:=\sigma(x_k)]}$ for all variables in
  $X'$. We call a polynomial $p$ with variables $X$ \emph{binary} iff $\forall
    \left(\sigma : X \rightarrow \{0, 1\}\right).~ \Pev\sigma p \in \{0, 1\}$.
  \label{def:poly}
\end{definition}

\begin{definition}[Arithmetization of {\gbcshort}s]
  Let $p[x_t/x_k]$ stand for the polynomial equal to $p$ with each occurrence of
  $x_k$ replaced by $x_t$, and $p \circledast q$ for the arithmetization of the
  binary boolean operator $\circledast$. The arithmetization
  ([\![~]\!]) of an {\gbcshort} node is defined inductively:
  \begin{align*}
     & [\![\top]\!] \equiv 1,  [\![\bot]\!] \equiv 0, [\![x_k]\!] \equiv x_k                                                             \\
     & [\![\varphi \circledast \psi]\!] \equiv [\![\varphi]\!] \circledast [\![\psi]\!], [\![\neg \varphi]\!] \equiv 1 - [\![\varphi]\!] \\
     & [\![\pi_{[x_k:=b]}\varphi]\!] \equiv \pi_{[x_k := [\![b]\!]]} [\![\varphi]\!] \\
     & [\![\varphi[x_t/x_k]]\!] \equiv [\![\varphi]\!][x_t/x_k]
  \end{align*}
  \label{def:arith}
\end{definition}

As an example, if \checker\ produces the following program trace: $g = x_1 \vee
x_2;\ f = g \oplus x_3$, it first constructs the {\gbcshort} $\varphi = (x_1
\vee x_2) \oplus x_3$. Then, it computes the {\gbcshort} $\conv(\varphi)$ by
adding degree-reduction gates, shown in Figure~\ref{fig:ebcfromlib}.

\begin{figure}[t]
  \centering
  \includegraphics[width=.8\textwidth]{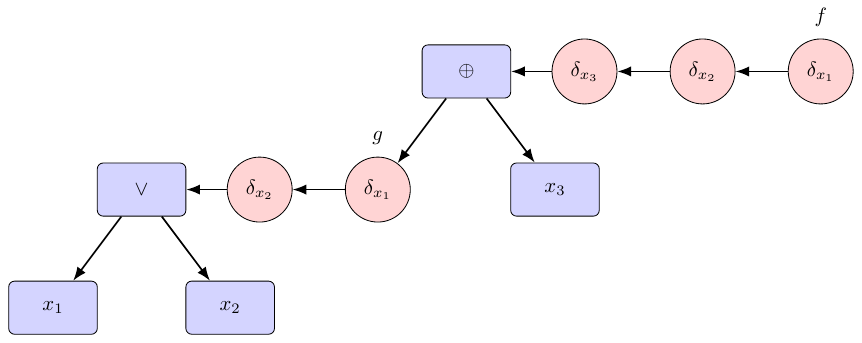}
  \caption{{\gbcshort} $\conv\left((v_1 \vee v_2) \oplus v_3\right)$generated by \checker.}
  \label{fig:ebcfromlib}
\end{figure}

\subsection{Generating Claims from Assertions}
\label{subsc:claimsgenerationapx}

The main paper introduces four types of claims that Prover makes about
\gbcshort\ nodes. The assertions present in the program trace given to \lib\ are
mapped to claims as follows, where $f, g$ are a program variable and $\psi,
\phi$ are the associated nodes in the circuit:
\begin{itemize}
  \item \texttt{assert $|\left\{\overline v \mid f(\overline v) = 1\right\}| =
  k$} (asserting number of sat.\ assignments) gets mapped to \emph{count} claim
  $\Sigma \fn(\psi) = k$

  \item \texttt{assert $f(b_1, \dotsc, b_n) = r$} (asserting evaluation result)
  gets mapped to \emph{$\mathbb{B}$-evaluation} claim $\Pev\sigma \fn(\psi) =
  r$, where $\forall i.\ \sigma(x_i) = b_i \in \{0, 1\}$.

  \item \texttt{assert $f = g$} and \texttt{assert $f \neq g$} (asserting
  functional-equality) get mapped to \emph{$\mathbb{B}$-equivalence} claims
  $\fn(\psi) = \fn(\phi)$ and $\fn(\psi) \neq \fn(\phi)$ respectively.
\end{itemize}

\section{The interactive proof protocol \prot}
\label{apx:prot}

We present {\prot} in a top-down manner. Section \ref{subsec:ebcertifytoplevelapx} describes the top-level procedure, which call procedures described in subsequent sections.

\subsection{\prot: Top-level Procedure}
\label{subsec:ebcertifytoplevelapx}

The main text explains that the top level procedure \prot\ first
uses \procname{Normalize} to turn all claims into \emph{$\F$-evaluation} claims.
It then iterates over the \gbcshort\ $\conv(\varphi)$ in topological order, engaging in a
\emph{round} at each gate. It does this via the \procname{PropagateAssignments}
procedure (Alg.~\ref{alg:certifyassignments}). In each iteration, it first
invokes \procname{Normalize} before one of the \procname{Propagate} variants
(explained below).

\begin{algorithm}[H]
  \caption{\texttt{\prot($\conv(\varphi), \mathcal{C}$)}}
  \label{alg:certify}
  \begin{algorithmic}
    \REQUIRE $\conv(\varphi)$ \COMMENT{\gbcshort\ generated from the trace}
    \REQUIRE $\mathcal{C}$ \COMMENT{Set of claims about output gates in $\conv(\varphi)$}
    \ENSURE $1$, if all claims could be certified, otherwise $0$
    \algrule
    \STATE $\mathcal{C} \leftarrow \texttt{Normalize($\mathcal{C}$)}$
    \COMMENT{$\mathcal{C}$ now only contains $\F$-evaluation claims} \RETURN
    $\texttt{CertifyAssignments($\conv(\varphi), \mathcal{C}$)}$
  \end{algorithmic}
\end{algorithm}

\subsubsection*{\procname{CertifyAssignments}} proceeds by starting a round for
each gate in topological order. It considers the set of assignment claims about
the current node and merges them. If the current gate is an \emph{input} gate,
Verifier computes the arithmetization of and checks whether the claim is correct
(`Decision' in main text). If it instead is an intermediate gate, it generates
new claims about the node's children, using one of the \texttt{Propagate}
subroutines, depending on the type of node, and updates $\mathcal{C}$ while
preserving CEHP.

The two subroutines used by \procname{CertifyAssignments}: \procname{Merge}
and \procname{Propagate} preserve CEHP with probabilities at least $1 - 2n/|\F|$
and $1 - 2/|\F|$ (shown below). We use these facts to prove soundness and
completeness.

\begin{lemma}
  For circuit $\conv(\varphi)$ and $\F$-equivalence claims $\Claimset$, if all
  claims in $\Claimset$ are true and Prover is honest,
  \procname{CertifyAssignments($\conv(\varphi), \Claimset$)} makes Verifier accepts
  with probability $1$ (\emph{completeness}). If at least one claim in
  $\Claimset$ is false, then Verifier rejects with probability at least $1 -
  4n|\varphi|/|\F|$.
  \label{lemma:certifyassignments}
\end{lemma}

\begin{proof}
  The proof largely follows from the correctness theorem of
  \CPCertify~\cite[Theorem~1]{CAV23}.

  The size of $\conv(\varphi)$ is at most $n|\varphi|$, with at most $|\varphi|$
  nodes that have more than one predecessor, (intermediate degree-reduction
  gates only have one predecessor). When iterating over gates in
  $\conv(\varphi)$, each node $\psi$ is only visited once, and at most
  $|\varphi|$ gates have more than one claim when visited, which is due to the
  limit of predecessors explained above, and the fact that initial claims cannot
  be about intermediate nodes added by $\conv$. This means that (i) claims in
  $\Claimset$ get replaced using \procname{Merge} at most $|\varphi|$ times. The
  \procname{Propagate} procedure is called for each node, at most
  $|\conv(\varphi)| \leq n|\varphi|$ times.

  Using the union bound, we have that if at least one $\F$-equivalence claim in
  $\Claimset$ is false, then a false claim is replaced by a true claim with probability at most
  $|\varphi| \cdot 2n/|\F| + n|\varphi| \cdot 2/|\F|$.
  Conversely, if all claims are true, then all claims added to $\Claimset$ are also true
  for an honest Prover.

  Thus, \procname{CertifyAssignments} causes Verifier to accept with probability
  $1$ given all claims are true and an honest Prover, and otherwise rejects with
  probability at least $1 - 4n|\varphi|/|\F|$.
\end{proof}

\begin{algorithm}[h]
  \caption{\texttt{CertifyAssignments($\conv(\varphi), \mathcal{C}$)}}
  \label{alg:certifyassignments}
  \begin{algorithmic}
    \REQUIRE $\conv(\varphi)$ \COMMENT{\gbcshort\ generated from the trace}
    \REQUIRE $\mathcal{C}$ \COMMENT{Set of $\F$-evaluation claims about nodes in $\conv(\varphi)$}
    \ENSURE $1$, if \texttt{Accept()} is called, $0$, if \texttt{Reject()} is called
    \algrule
    \FOR{Gate $\psi$ in topological order of $\conv(\varphi)$}
    \STATE $C \leftarrow \mathcal{C}(\psi)$ \COMMENT{Get the set of assignment claims about node $\psi$}
    \STATE $(\Pev\sigma [\![\psi]\!] = k) \leftarrow \texttt{Merge($C, \psi$)}$
    \IF{\texttt{IsLeaf($\psi$)}}
    \STATE $t \leftarrow \Pev\sigma [\![\psi]\!]$ \COMMENT{Decision: $\psi$ is an input gate that Verifier checks in linear time}
    \IF{$t \neq k$}
    \STATE \texttt{Reject()}
    \ENDIF
    \ELSE
    \STATE $C' \leftarrow \texttt{Propagate($\psi, \Pev\sigma [\![\psi]\!] = k$)}$
    \STATE $\mathcal{C} \leftarrow \mathcal{C} \cup C'$
    \ENDIF
    \ENDFOR
    \STATE $\texttt{Accept()}$
  \end{algorithmic}
\end{algorithm}

\spnewtheorem*{L1}{Lemma~\ref{lemma:ebcertifycorrectness}}{\bfseries \upshape}{\itshape}

Using our knowledge about the CEHP properties of \procname{Normalize}
(Lemma~\ref{lemma:normalizecehp}) and previous Lemma, proving the original
correctness statement of the main text is straightforward:

\begin{L1}
 If $\mathcal{C}$
  contains a wrong claim about the \gbcshort\ $\conv(\varphi)$ with $n$ variables,
  {\prot} will accept with probability at most $\left({4n |\varphi| +
      n}\right)/{\F}$ for any prover (soundness). If all claims in $\mathcal{C}$
  are correct, it will accept with probability $1$ given an honest prover
  (completeness).
\end{L1}

\begin{proof}
  Using the union bound, we have that \procname{Normalize} preserves CEHP with
  probability at least $1 - n/|\F|$, and that \procname{CertifyAssignments} has
  a soundness error of at most $4n|\varphi|/|\F|$, and perfect completeness.
\end{proof}

\subsection{Initialization: Normalize}
\label{subsc:normalizeapx}

The pseudocode for the initialization procedure \procname{Normalize} explained
in Section~\ref{ssubsec:clch} is shown in Alg.~\ref{alg:normalize}. The main
text explains that \procname{Normalize} replaces all four claim types by
$\F$-evaluation claims in a way that preserves CEHP. Conceptually,
this is because arithmetizations encode more information about a function than
just its value under a boolean assignment. Namely, assigning each variable to
$2^{-1} \in \F$, instead of a boolean value, evaluates
$\Eval{\varphi}$ to the number of satsifying assignments
(\cite[Lemma~2]{CAV23}). Further, because each function has a
\emph{unique} arithmetization (\cite[Prop.~3]{CAV23}), iff $\fn(\varphi)
= \fn(\psi)$, then $\Eval{\varphi} = \Eval{\psi}$.

\begin{lemma}
 If two degree-reduced \gbcshort\ nodes $\varphi$, $\psi$ represent the same boolean function, then
 \[
   \Pr_\sigma\left[\Pev\sigma([\![\varphi]\!] - \Eval{\psi}) = 0\right] = 1
 \]
 Conversely, if they do not represent the same boolean function, then
 \[
   \Pr_\sigma\left[\Pev\sigma([\![\varphi]\!] - \Eval{\psi}) \neq 0\right] \geq 1 - \frac{n}{\F}
 \]\label{lemma:equivcorrectness}
\end{lemma}
\begin{proof}
  We need to show that the probability of error is $\leq
    \frac{n}{|\F|}$:

  We have $\varphi \equiv \psi$ iff $[\![\phi]\!] - \Eval{\psi} = p - q = 0$ by
  \cite[Prop.~3]{CAV23} (arithmetization produces a
  \emph{unique} polynomial). We further have that the total degree $d$ of both
  $p$ and $q$ is $\leq n$, because there are at most $n$ variables of degree $1$
  in any monomial. By the Schwartz-Zippel lemma:
  \begin{align*}
    p - q = 0    & \Longrightarrow \Pr_{\sigma}\left[\Pev{\sigma}(p - q) = 0\right] = \Pr_\sigma\left[\Pev\sigma p - \Pev\sigma q = 0\right] = 1                      \\
    p - q \neq 0 & \Longrightarrow \Pr_\sigma\left[\Pev\sigma(p - q) = 0\right] = \Pr_\sigma\left[\Pev\sigma p - \Pev\sigma q = 0\right] \leq  \frac{d}{|\F|}
  \end{align*}
  The first goal, when $\varphi \equiv \psi$, immediately follows.
  We conclude our second goal, when $\varphi \not \equiv \psi$, by
  \[
    \Pr_\sigma\left[\Pev\sigma([\![\varphi]\!] - \Eval{\psi} \neq 0)\right] \\
    = 1 - \Pr_\sigma\left[\Pev\sigma([\![\varphi]\!] - \Eval{\psi} = 0)\right] \\
    \geq 1 - \frac{d}{\F} \geq 1 - \frac{n}{\F}
  \]
\end{proof}

\begin{algorithm}[h]
  \caption{\texttt{Normalize($\mathcal{C}$)}}
  \label{alg:normalize}
  \begin{algorithmic}
    \REQUIRE $\mathcal{C}$ \COMMENT{Set of all four possible types of claims}
    \ENSURE A set of assignment claims reduced from claims in $\mathcal{C}$.
    \algrule
    \STATE $\mathcal{C}' = \{\}$
    \FORALL{$c \in \mathcal{C}$}
    \IF{$c = (\Pev\sigma [\![\varphi]\!] = k)$}
    \STATE $\mathcal{C}' \leftarrow \mathcal{C}' \cup \{c\}$
    \ELSIF {$c = (\Pev\sigma \fn(\varphi) = b)$}
    \STATE $\mathcal{C}' \leftarrow \mathcal{C}' \cup \{\Pev\sigma [\![\varphi]\!] = b\}$
    \ELSIF {$c = (\Sigma \fn(\varphi) = k)$}
    \STATE $\forall {x_i \in \Free(\varphi)}.~ \sigma(x_i) = 2^{-1}$
    \STATE $\mathcal{C}' \leftarrow \mathcal{C}' \cup \{\Pev\sigma [\![\varphi]\!] = k \cdot 2^{-(|\Free(\varphi)|)}\}$
    \ELSIF {$c = (\varphi \equiv \psi = 1)$}
    \STATE $\sigma \R ([x_1, \dots, x_n) \rightarrow \F)$
    \STATE $p \leftarrow \texttt{Challenge($\Pev\sigma \varphi$)},\ q \leftarrow \texttt{Challenge($\Pev\sigma \psi$)}$
    \IF {$p \neq  q$}
    \STATE $\texttt{Reject()}$
    \ENDIF
    \STATE $\mathcal{C}' \leftarrow \mathcal{C}' \cup \{\Pev\sigma [\![\varphi]\!] = p, \Pev\sigma \Eval{\psi} = q\}$
    \ELSIF {$c = (\varphi \equiv \psi = 0)$}
    \STATE $\sigma \leftarrow \texttt{ChallengeDistinct($\varphi, \psi$)}$
    \STATE $p \leftarrow \texttt{Challenge($\Pev\sigma \varphi$)},\ q \leftarrow \texttt{Challenge($\Pev\sigma \psi$)}$
    \IF {$p =  q$}
    \STATE $\texttt{Reject()}$
    \ENDIF
    \STATE $\mathcal{C}' \leftarrow \mathcal{C}' \cup \{\Pev\sigma [\![\varphi]\!] = p, \Pev\sigma \Eval{\psi} = q\}$
    \ENDIF
    \ENDFOR
    \RETURN $\mathcal{C}'$
  \end{algorithmic}
\end{algorithm}

Now we can prove that \procname{Normalize($\Claimset$)} is CEHP with probability
$1 - n/|\F|$ for an {\gbcshort} with $n$ variables.

\begin{lemma}
  For claims relating to an \gbcshort with $n$ variables, \procname{Normalize($\Claimset$)} is CEHP with probability
  at least $1 - n/|\F|$.
  \label{lemma:normalizecehp}
\end{lemma}

\begin{proof}
  \procname{Normalize} replaces $\mathbb{B}$-evaluation and count-claims without
  any error probability -- i.e.\ for those two cases, CEHP is preserved with
  probability $1$ (\cite[Prop.~1, Lemma~2]{CAV23}). The only possibility
  to replace a wrong claim by a true claim is for $\mathbb{B}$-equivalence claims
  of the form $\fn(\psi_1) - \fn(\psi_2)$ (Lemma~\ref{lemma:equivcorrectness}).

  With these facts we show that \procname{Normalize} is CEHP according to the following disjoint cases:
  \begin{align*}
     & \text{If $\mathcal{C}$ contains a \emph{false} claim of the form $\varphi = \psi$ then}  \\
     & \hspace{4em} \forall P.\ \Pr[\texttt{Normalize($\mathcal{C}$)} \text{ contains $0$ false  claims}] \leq \frac{n}{\F} \\
     & \text{If $\mathcal{C}$ contains a \emph{false} claim \emph{not} of the form $\varphi = \psi$ then               } \\
     & \hspace{4em}\forall P.\ \Pr[\texttt{Normalize($\mathcal{C}$)} \text{ contains $0$ false claims}] = 0                        \\
     & \text{If $\mathcal{C}$ contains $0$ false claims then}  \\
     & \hspace{4em}\exists P.\ \Pr[\texttt{Normalize($\mathcal{C}$)} \text{ contains $0$ false claims}] = 1                        \\
  \end{align*}
\end{proof}

\subsection{Round for Gate}
\label{subsc:roundapx}

Recall that \prot\ checks all remaining $\F$-evaluation claims by iterating over
each gate in a circuit in topological order, propagating claims to leaves via
\procname{CertifyAssignments}. This section explains of processing an individual
gate.

\subsubsection*{\procname{Merge}:} \prot\ invokes
\procname{Merge} to generate a single claim from a set of claims for
each node. As mentioned in the paper, \procname{Merge} is already used by
\CPCertify, and the proof that it preserves CEHP largely follows from
\cite[Prop.~2]{CAV23}. For completeness, we explain the procedure.
\procname{Merge($C, \psi$)} is invoked at the beginning of Round for gate
$\psi$, $C \subseteq \Claimset$ being all $\F$-evaluation claims about gate
$\psi$. As mentioned, if $\psi$ has $k$ predecessors, then $C$ contains $k$
claims in addition to any initial claims.

The function iterates over every variable $x_i$ of the assignments and, for each
claim, sends Prover $\procname{Challenge($\Pev{\sigma \setminus v_i}
\Eval{\psi}$)}$, to which it responds with a linear polynomial undefined in
$x_i$. Verifier runs sanity checks on the answer. If the polynomial is not
congruent with the original claim, Verifier rejects. Otherwise, it replaces all
claims by sampling a value for variable $x_i$ from $\F$ uniformly at random.
After going over all variables, the claims should be equal, otherwise Verifier
rejects. Verifier returns the single remaining claim.

\begin{algorithm}[h]
  \caption{\texttt{Merge($C, \psi$)}}
  \label{alg:merge}
  \begin{algorithmic}
    \REQUIRE $\psi$ \COMMENT{A node in the \gbcshort\ $\conv(\varphi)$}
    \REQUIRE $C$ \COMMENT{A set of assignment claims about $\psi$}
    \ENSURE A single assignment claim about $\psi$, reduced from $C$
    \algrule
    \FORALL{$x_i \in \Free(\psi)$}
    \STATE $C' \leftarrow \{\}$
    \FORALL{$(\Pev\sigma [\![\psi]\!] = k) \in C$}
    \STATE $p(x_i) \leftarrow \texttt{Challenge($\Pev{\sigma \setminus x_i}[\![\psi]\!]$)}$
    \IF{$p(\sigma(x_i)) \neq k$}
    \STATE \texttt{Reject()}
    \ENDIF
    \STATE $C' \leftarrow C' \cup (\Pev{\sigma \setminus x_i} [\![\psi]\!] = p(x_i))$ \COMMENT{Save polynomial in $C'$}
    \ENDFOR
    \STATE $r \R \F$
    \STATE $C \leftarrow \left\{\Pev{\sigma[x_i \rightarrow r]} [\![\psi]\!] = p(r) \mid \Pev{\sigma} [\![\psi]\!] = p(x_i) \in C' \right\}$
    \ENDFOR
    \IF{$|C| \neq 1$}
    \STATE \texttt{Reject()}
    \ENDIF
    \RETURN $\mathit{claim} \in C$
  \end{algorithmic}
\end{algorithm}

\procname{Merge} maintains CEHP with probability $1 - 2k/|\F|$ if $\Eval{\psi}$ has $k$
free variables, where $k \leq n$ ($n$ being the number of variables present
$\psi$'s \gbcshort).
\begin{lemma}
  \procname{Merge($C, \psi$)} preserves CEHP with probability at least $1 - 2n/|\F|$.
\end{lemma}
The proof is part of \cite[Prop.~2]{CAV23}. See in particular the
reasoning for step (b.1).

\subsubsection*{\texttt{Propagate} [Binary Operation]:} As explained in
Section~\ref{subsubsec:round}, propagating claims about \emph{binary operation}
nodes was already present in \CPCertify.
Alg.~\ref{alg:certify_propagatebin} shows its pseudocode.
The procedure preserves CEHP with no possible error:
\begin{lemma}
  \procname{Propagate($\varphi = \psi_1 \circledast \psi_2$, $\Pev{\sigma} \Eval{\varphi} = k$)}
  preserves CEHP with probability $1$.
\end{lemma}
\begin{proof}
  Consider the claim $\Pev{\sigma} \Eval{\varphi} = k$ to be true. In that case, honest Prover can
  provide correct answers to \procname{Challenge}, s.t.\ the two new claims remain true.

  Instead consider the claim to be false. In that case, Prover must lie by
  answering with at least one wrong value $p_i \in \F =
  \procname{Challenge($\Pev{\sigma} \Eval{\psi_i}$)}$. This necessitates that
  the added claim $\Pev{\sigma} \Eval{\psi_i} = p_i$ is also wrong, preserving
  CEHP.
\end{proof}

\begin{algorithm}[h]
  \caption{\texttt{Propagate($\psi_1 \circledast \psi_2$, $\Pev\sigma [\![\psi_1 \circledast \psi_2]\!] = k$)}}
  \label{alg:certify_propagatebin}
  \begin{algorithmic}
    \STATE $p \leftarrow \texttt{Challenge($\Pev\sigma [\![\psi_1]\!]$)}$
    \STATE $q \leftarrow \texttt{Challenge($\Pev\sigma \Eval{\psi_2}$)}$
    \IF{\texttt{$p \circledast q$} $\neq k$}
    \STATE \texttt{Reject()}
    \ENDIF
    \RETURN $\{\Pev\sigma [\![\psi_1]\!] = p, \Pev\sigma \Eval{\psi_2} = q\}$
  \end{algorithmic}
\end{algorithm}

\subsubsection*{\texttt{Propagate} [Degree Reduction]:} Also taken from
\CPCertify, this case can cause Verifier to violate CEHP with small probability.
The pseudocode of the procedure is shown in
Alg.~\ref{alg:certify_propagatedegree}. The reason we can bound the probability
to at most $2/|\F|$ is explained in \cite[Prop.~2]{CAV23}. Essentially,
because of the Schwartz-Zippel Lemma, any polynomial of degree $i$ has at most
$i$ roots, which means that if we sample u.a.r from $\F$, the chance of sampling
a root is at most $i/|\F|$.

\begin{lemma}
  \procname{Propagate($\delta_{x_i} \psi, \Pev{\sigma} \Eval{\delta_{x_i} \psi} = k$)}
  preserves CEHP with probability at least $1 - 2/|\F|$.
\end{lemma}
\begin{proof}
  If the claim $\Pev{\sigma} \Eval{\delta_{x_i} \psi} = k$ is correct, then honest Prover
  answers with the correct polynomial $p(x_i) = \Pev{\sigma \setminus x_i} \Eval{\psi}$,
  so the added claim remains true.

  If the claim is false, then $x_i \cdot \Pev{\sigma \setminus x_i} \Eval{\psi}
  + (1 - x_i) \cdot \Pev{\sigma \setminus x_i} \Eval{\psi}$ is \emph{not} equal
  to $k$. So Prover must supply a wrong polynomial $p(x_i)$ unless Verifier
  immediately rejects the claim. Due to the Schwartz-Zippel Lemma, the added
  claim $\Pev{x_i \rightarrow r} \Eval{\psi} = p(r)$ remains false with
  probability at least $1 - 2/|\F|$.
\end{proof}

\begin{algorithm}[h]
  \caption{\texttt{Propagate($\delta_{x_i} \psi$, $\Pev\sigma [\![\delta_k \psi]\!] = k$)}}
  \label{alg:certify_propagatedegree}
  \begin{algorithmic}
    \STATE $p(x_i) \leftarrow \texttt{Challenge($\Pev{\sigma \setminus x_i} \Eval{\psi}$)}$
    \STATE $q(x_i) = x_i \cdot p(1) + (1 - x_i) \cdot p(0)$
    \IF{$q(\sigma(x_i)) \neq k$}
    \STATE \texttt{Reject()}
    \ENDIF
    \STATE $r \R \F$
    \RETURN $\{ \Pev{\sigma[x_i \rightarrow r]}\Eval{\psi} = p(r) \}$
  \end{algorithmic}
\end{algorithm}

\subsubsection*{\texttt{Propagate} [Projection]:} Also part of \CPCertify,
propagating claims about projection gates trivially preserves CEHP with
probability $1$. For completeness, pseudocode is presented in
Alg.~\ref{alg:certify_propagateproject}.

\begin{algorithm}[h]
  \caption{\texttt{Propagate($\pi_{[x_i := b]} \psi$, $\Pev\sigma [\![\pi_{[x_i := b]} \psi]\!] = k$)}}
  \label{alg:certify_propagateproject}
  \begin{algorithmic}
    \RETURN $\{\Pev{\sigma[x_i \rightarrow [\![b]\!]]}\psi = k\}$
  \end{algorithmic}
\end{algorithm}

\subsubsection*{\texttt{Propagate} [Renaming]:} This case for a round of \prot\
is new, as \CPCertify\ did not previously support renaming gates. The procedure
also preserves CEHP with no possibility for error. This is ensured by Verifier
and \lib\ statically on the \gbcshort\ it generates from a trace of \checker. The idea
is that renaming $x_i$ to $x_t$ via a renaming gate $\psi[x_t/x_i]$ is only allowed
if the variable $x_t$ is not in $\Free(\psi)$, which means that $\Eval{\psi}$
has degree $0$ for variable $x_t$.

\begin{lemma}
  \procname{Propagate($\psi[x_t/x_i], \Pev{\sigma} \Eval{\psi[x_t/x_i]} = k$)}
  preserves CEHP with probability $1$.
\end{lemma}
\begin{proof}
  If the claim $\Pev{\sigma} \Eval{\psi[x_t/x_i]} = k$ is true, then
  $\Eval{\psi[x_t/x_i]} = \Eval{\psi}[x_t/x_i]$. Because $\Eval{\psi}$ does not
  contain variable $x_t$, $\Eval{\psi} = \Eval{\psi[x_t/x_i]}[x_i/x_t]$ and the
  claim $\Pev{\sigma[x_i \rightarrow \sigma(x_t)]} \Eval{\psi}$ remains true.

  Conversely, if the claim is false, then due to the same reasoning the added
  claim remains false.
\end{proof}

\begin{algorithm}[h]
  \caption{\texttt{Propagate($\phi[x_t/x_i]$, $\Pev\sigma [\![\phi[x_t/x_i]]\!] = k$)}}
  \label{alg:certify_propagaterename}
  \begin{algorithmic}
    \STATE{$\sigma' = \sigma[x_i \rightarrow \sigma(x_t)]$}
    \RETURN $\{\Pev{\sigma'} \phi = k\}$
  \end{algorithmic}
\end{algorithm}

\section{Generating eBDDs: \procname{ApplyEBDD} vs.\ \procname{ComputeEBDDs}}
\label{apx:newname}

In the paper we have given a high-level introduction to our improvements to \blic's \procname{ComputeEBDDs}
procedure, effectively allowing for a global unique BDD table. In this
section we present both \procname{ComputeEBDDs} and our improved version \procname{ApplyEBDD}.

\subsection{The algorithm \procname{ComputeEBDDs}}

\begin{figure}[ht]
  \includegraphics[width=\textwidth]{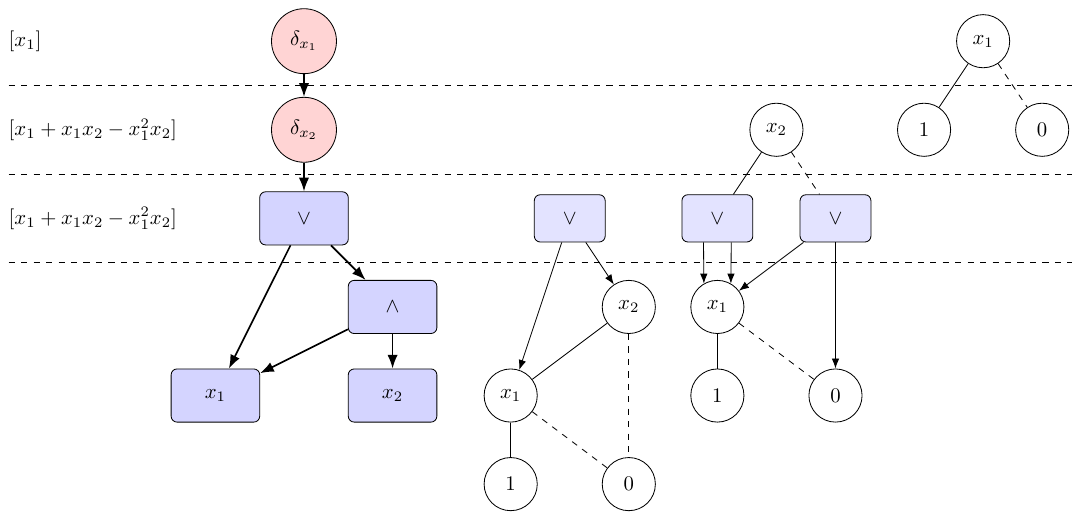}
  \caption{GBC and arithmetization of nodes (left) and eBDDs with equal arithmetization (right).}
  \label{fig:ebdd}
\end{figure}

\cite{CAV23} introduced the data structure of \emph{extended BDDs}.
Figure~\ref{fig:ebdd} illustrates an example of how eBDDs are used to encode
polynomials of an underlying circuit. On the left is the \gbcshort\ $\varphi =
\delta_1 \delta_2 \psi$ with $\psi =  v_1 \vee \left(v_1 \wedge v_2\right)$
($\conv$ would also add degree reductions to $v_1 \wedge v_2$). To the right of
the \gbcshort\ are three eBDDs for $\Eval{\psi}$, $\Eval{\delta_2 \psi}$, and
$\Eval{\delta_1 \delta_2 \psi}$. The polynomial equal to both each \gbcshort\ node and
its eBDD is displayed in brackets on the left.

\begin{figure}[ht]
  \centering
  \includegraphics[width=\textwidth]{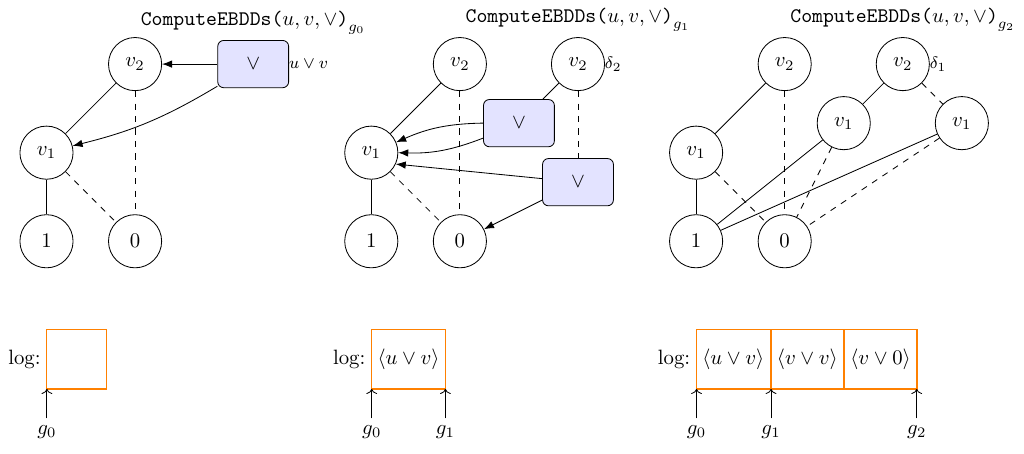}
  \caption{Previous computation of eBDDs using mutable nodes and a log.}
  \label{fig:ebddcompute}
\end{figure}

For each gate $\psi = \psi_1 \circledast \psi_2$ in a circuit $\conv(\varphi)$
with BDDs $u_1,\ u_2$ for $\psi_1,\ \psi_2$, algorithm $\procname{ComputeEBDDs}$
computed all eBDDs for polynomials in the sequence $g_0 = \Eval{\psi},\
g_1 = \delta_n \Eval{\psi},\ \cdots,\ g_n = \delta_1 \delta_2
\dotsb \Eval{\psi}$ in the same time as \procname{Apply($u_1, u_2,
\circledast$)}. As state in section~\ref{subsec:shortcomings}, this was achieved
using in-place mutation and an \emph{undo-log}. Prover would access each of the
$g_i$ eBDDs by undoing modifications from the log. Figure~\ref{fig:ebddcompute}
sketches the eBDDs produced by \procname{ComputeEBDDs}. The figure shows the
sequence of eBDDs $g_0 = \Eval{\psi_1 \vee \psi_2},\ g_1 = \delta_2 g_0,\ g_2 =
\delta_1 g_1$. As Prover modifies nodes, it writes the modifications to its
undo-log. To access any eBDD $g_i$ in the sequence, Prover simply undoes the
modifications from the log up to its index.

Using standard BDD algorithms, \cite{CAV23}'s Prover could evaluate each
BDD's polynomial $\Eval{g}$ in linear time $|g|$. In fact, evaluating all
polynomials $\Eval{g_0}, \dotsc, \Eval{g_n}$ of a sequence corresponding to
degree reductions could \emph{also} be done in time $|g|$ (as opposed to $n|g|$)
using a cache. The details of said cache are technical, but essentially exploits
the fact that the total number of modifications between the BDDs are of size
$|g|$.

\subsection{The algorithm \procname{ApplyEBDD}}

To prove the complexity and correctness of our novel method for computing eBDDs
as stated in Proposition~\ref{prop:computebdd}, we first introduce eBDD
\emph{invariants} that \procname{ApplyEBDD} maintains. Recall from the main text
that the three types of eBDD nodes by \lib\ are \emph{final}, \emph{standard},
and \emph{binary operation} nodes.

\begin{definition}[\procname{ApplyEBDD($f, g, \circledast$)} {[Invariant]}] The
 returned eBDD $r$ and any child $w$ adhere to the following invariants in
 addition to BDD invariants (see.\ Appendix~\ref{apx:bdds}):
 \begin{enumerate}
   \item $[\![r]\!] = [\![f]\!] \circledast [\![g]\!]$.
   \item If $w$ is a binary operation node, all its children are final BDD
         nodes.
    \item $\link(w) = \langle v_k, l', r'\rangle$ is either a standard eBDD or
    final BDD node that fulfills $v_k = \var(w)$ and $[\![\link(w)]\!] = v_k
    \cdot \pi_{[v_k := 1]} [\![w]\!] + (1 - v_k) \cdot \pi_{[v_k := 0]}
    [\![w]\!]$.
   \item $\link(\link(w)) = u$ is a \emph{final} BDD. We call BDDs final if
         they do not contain any binary operation nodes and thus represent
         standard BDDs. $[\![u]\!] = \delta_1 \dotsb \delta_n [\![r]\!]$ for $n$
         being the number of variables in $[\![w]\!]$.
   \item $w$ is a final BDD iff $\link(w) = w$.
 \end{enumerate}
 \label{def:invariants}
\end{definition}

Next, we prove using induction on \procname{ApplyEBDD}'s recursive call graph
that it maintains eBDD invariants.

\begin{lemma}[\procname{ApplyEBDD($f, g, \circledast$)} {[correctness]}] Given
 final eBDD arguments $f, g$ that adhere to the invariants
 (Def.~\ref{def:invariants}), the result $r = \procname{ApplyEBDD($u, w,
 \circledast$)}$ also adheres to said invariants.
 \label{lemma:computeebdds}
\end{lemma}

\begin{proof}
  We proceed by f-induction on \texttt{ComputeEBDDs} (termination is proven by
  lemma~\ref{lemma:complexity}). Next we analyze the possible forms that the
  argument BDDs $f$ and $g$ may have according to the following cases:

  \paragraph{Base $f \in \{0, 1\}$ or $g \in \{0, 1\}$}: In the base case Prover
    computes the eBDD $r = f \circledast g$ directly according to the truth
    table of boolean operation $\circledast$. The result is either the final BDD
    node $1$ or $0$, which fulfills $\Eval{r} = \Eval{f} \circledast \Eval{g}$
    (first invariant). While not explicitly stated in the pseudocode, Prover
    also sets $\link(\link(r)) = \link(r) = r$, such that invariants three to
    five hold. The second invariant is vacuously true.

  \paragraph{Step $f = \langle f_i, f_l, f_r \rangle$ and $g = \langle g_i, g_l,
    g_r \rangle$}: \procname{ApplyEBDD} first projects both $f$ and $g$ according
    to variable $v_i = \max(f_i, g_i)$: $f_l = f|_{v_i \leftarrow 0}, f_r =
    f|_{v_i \leftarrow 1}, g_l = g|_{v_i \leftarrow 0}, g_r = g|_{v_i \leftarrow
    1}$. It then computes new child BDD nodes $l$ and $r$ via recursion. We use
    the induction hypothesis to derive that $l$ and $r$ adhere to the eBDD
    invariants, giving us:
  \begin{align*}
    [\![l]\!]               & = \pi_{[v_i := 0]} [\![f]\!] \circledast \pi_{[v_i := 0]} [\![g]\!]                                 \\
    [\![r]\!]               & = \pi_{[v_i := 1]} [\![f]\!] \circledast \pi_{[v_i := 1]} [\![g]\!]                                 \\
    [\![\link(l)]\!]        & = \var(l) \cdot \pi_{[\var(l) := 1]} [\![l]\!] + (1 - \var(l)) \cdot \pi_{[\var(l) := 0]} [\![l]\!] \\
    [\![\link(r)]\!]        & = \var(r) \cdot \pi_{[\var(r) := 1]} [\![r]\!] + (1 - \var(r)) \cdot \pi_{[\var(r) := 0]} [\![r]\!] \\
    [\![\link(\link(l))]\!] & = \delta_1 \dots \delta_{\var(l)} [\![l]\!]                                                         \\
    [\![\link(\link(r))]\!] & = \delta_1 \dots \delta_{\var(r)} [\![r]\!]                                                         \\
  \end{align*}

  The returned eBDD $b$ inherently adheres to the first invariant since
  $\Eval{b} = \Eval{f} \circledast \Eval{g}$. $\link(b) = \mathit{node} =
  \langle v_i, l, r \rangle$, which is a standard eBDD that fulfills the following
  equation, thus satisfying invariants two and three.
  \begin{multline*}
    [\![\mathit{node}]\!] = v_i \cdot [\![r]\!] + (1 - v_i) \cdot [\![l]\!] \\
    = v_i \cdot (\pi_{[v_i := 1]} [\![u]\!] \circledast \pi_{[v_i:=
      1]} [\![w]\!]) + (1 - v_i) \cdot (\pi_{[v_i := 0]} [\![u]\!] \circledast
    \pi_{[v_i:= 0]} [\![w]\!])\\
    = v_i \cdot (\pi_{[v_i := 1]} ([\![u]\!] \circledast [\![w]\!]))
    + (1 - v_i) \cdot (\pi_{[v_i := 0]} ([\![u]\!] \circledast
      [\![w]\!])) \\
    = v_i \cdot (\pi_{[v_i := 1]} [\![b]\!]) + (1 - v_i) \cdot (\pi_{[v_i := 0]} [\![b]\!])
  \end{multline*}

  The last two invariants hold because $\link(\mathit{node}) = \mathit{final}$ and
  $\mathit{final} = \langle v_i, \link(\link(l)), \link(\link(r))\rangle$ fulfills:
  \begin{multline*}
    [\![\mathit{final}]\!] = v_i \cdot (\delta_1 \dots \delta_{\var(r)} [\![r]\!]) + (1 - v_i) \cdot (\delta_1 \dots \delta_{\var(l)} [\![l]\!])\\
    \stackrel{1}{=} v_i \cdot (\delta_1 \dots \delta_{i - 1} \pi_{[v_i := 1]} [\![b]\!]) + (1 - v_i) \cdot (\delta_1 \dots \delta_{i - 1} \pi_{[v_i := 0]} [\![b]\!]) \\
    = \delta_1 \dots \delta_k [\![b]\!]
  \end{multline*}
  Where equation $(1)$ uses the facts that eBDDs are ordered, $\var(r) < v_i$
  and $\var(l) < v_i$. Note that \texttt{Reduce} does not change a node's
  arithmetization (it either returns the node or its child, if both children are
  equal). 

  The remaining structural invariants of standard BDDs are guaranteed by the
  induction hypothesis and \procname{Reduce} in the same manner as in the
  well-known \procname{Apply} procedure.
\end{proof}

Figure~\ref{fig:applyandnew} shows that \procname{ApplyEBDD} follows the same
recursive structure as \procname{Apply}. Because the arguments given by Prover
are final BDDs, it is easy to bound the runtime and space usage of our new
algorithm according to the well-known bounds of \procname{Apply}.

\begin{lemma}[\procname{ApplyEBDD($f, g, \circledast$)} {[complexity]}] Given
 final BDDs $f, g$, if the results of all recursive invocations of
 \procname{ApplyEBDD} are cached according to their arguments, then (a)
 \procname{ApplyEBDD($f, g, \circledast$)} takes the same time as
 \procname{Apply($f, g, \circledast$)} up to a constant factor, and (b) creates
 at most $3 \cdot |f||g|$ eBDD nodes.
 \label{lemma:complexity}
\end{lemma}

\begin{proof}
  From the fact that $f, g$ are final eBDD nodes we have that they comply to all
  invariants of ordered and reduced BDD nodes. It is obvious that the recursion
  tree of \texttt{ApplyEBDD($f, g, \circledast$)} is equal to that of
  \texttt{Apply($f, g, \circledast$)}. Also, each recursion of \texttt{ApplyEBDD}
  (excluding recursive calls) takes constant time, yielding proposition (a). 
  For (b) note from (a) that we have an upper bound of $|u||w|$ recursive
  invocations, and in each invocation we create at most three nodes:
  $\mathit{node},\ \mathit{final}$, and $b$.
\end{proof}

\section{Implementing \procname{Challenge} and \procname{ChallengeDistinct}}
\label{apx:challengeprover}

\begin{algorithm}[ht]
  \caption{\texttt{AnswerChallenge($w, \sigma, k$)}}
  \label{alg:answerchallenge}
  \begin{algorithmic}
    \REQUIRE $w$ \COMMENT{An eBDD calculated by \procname{ApplyEBDD}}
    \REQUIRE $\sigma$ \COMMENT{Partial assignment}
    \ENSURE Polynomial $\Pev\sigma \left(\delta_k \dotsb \delta_{\var(w)} [\![w]\!] \right)$
    \algrule
    \IF{$w \in \{0, 1\}$}
    \RETURN $[\![w]\!]$
    \ENDIF
    \IF{$w = \langle u \circledast v \rangle$ \AND $\var(w) \geq k$}
    \STATE $w \leftarrow \link(w)$
    \ELSIF{$w = \langle u \circledast v \rangle$}
    \RETURN $[\![\texttt{AnswerChallenge($u, \sigma, k$)} \circledast \texttt{AnswerChallenge($v, \sigma, k$)}]\!]$
    \ENDIF
    \STATE $\langle v_i, l, r \rangle = w$
    \IF{$\sigma(v_i) = c$}
    \RETURN $c \cdot \texttt{AnswerChallenge($r, \sigma, k$)} + (1 - c) \cdot \texttt{AnswerChallenge($l, \sigma, k$)}$
    \ELSE
    \RETURN $v_i \cdot \texttt{AnswerChallenge($r, \sigma, k$)} + (1 - v_i) \cdot \texttt{AnswerChallenge($l, \sigma, k$)}$
    \ENDIF
  \end{algorithmic}
\end{algorithm}

The main text states that, given two GBC nodes $\psi_1, \psi_2$ with BDDs $u_1,
u_2$, Prover computes the series of eBDDS $\Eval{g_0} = \Eval{\psi_1 \circledast
\psi_2},\ \Eval{g_{i + 1}} = \delta_{n - i} \dotsb \delta_n \Eval{\psi_1
\circledast \psi_2}$ using $\procname{ApplyEBDD($u_1, u_2, \circledast$)}$ and
uses them to answer challenges sent by Verifier. There are two types of
challenges Prover needs to answer: partial evaluation challenges and distinct
assignment challenges. In both cases Prover traverses BDDs corresponding to the
requested arithmetizations. Each BDD $g_{i + 1}$ can be accessed from the eBDD
root, returned by \procname{ApplyEBDD}, by interpreting each node $u$ with
$\var(u) > n - i$ as $\link(u)$, and $g_0 = \link(\link(w))$. This strategy is
implemented in \procname{AnswerChallenge($u, \sigma, k$)}
(Alg.~\ref{alg:answerchallenge}).

\begin{lemma}[\texttt{AnswerChallenge} {[Correctness]}] Given an eBDD $w$ adhering
 to the invariants of def.~\ref{def:invariants} and $0 \leq k$,
 $\procname{AnswerChallenge($w, \sigma, k$)} = \Pev\sigma \delta_k \dotsb
 \delta_{\var(w)} [\![w]\!]$.
 \label{lemma:evaluateebddcorrectness}
\end{lemma}

\begin{proof}
  By induction on $w$ according to the inductive structure of eBDDs.

  \paragraph{Base $w \in \{0, 1\}$}: The conclusion follows directly.

    \paragraph{Step $w = \langle u \circledast v \rangle$}:
    \begin{itemize}
      \item Case $\var(w) \geq k$:
            \begin{align*}
              \langle v_j, l, r \rangle             & = \link(w)                                                                                                                \\
              v_j                                   & = \var(w)                                                                         &  & (\text{\ref{def:invariants}}) \\
              [\![\link(w)]\!]                      & = v_j \cdot [\![r]\!] + (1 - v_j) \cdot [\![l]\!]                                                                         \\
                                                    & = v_j \cdot \pi_{[v_j := 1]} [\![w]\!] + (1 - v_j)\cdot \pi_{[v_j := 0]}[\![w]\!] &  & (\text{\ref{def:invariants}}) \\
              \texttt{AnswerChallenge($l, \sigma, k$)} & = \Pev{\sigma[v_j \rightarrow 0]} \delta_k \dotsb \delta_{\var(l)} [\![w]\!]      &  & (\text{IH})                        \\
              \texttt{AnswerChallenge($r, \sigma, k$)} & = \Pev{\sigma[v_j \rightarrow 1]} \delta_k \dotsb \delta_{\var(r) } [\![w]\!]     &  & (\text{IH})
            \end{align*}
            Thus, using the fact that $v_j = \var(w) \geq \max(\var(l), \var(r)) \geq k$, if $v_j \notin \sigma$:
            \begin{multline*}
              v_j \cdot \texttt{AnswerChallenge($r, \sigma, k$)} + (1 - v_j) \cdot \texttt{AnswerChallenge($l, \sigma, k$)} \\
              = \Pev\sigma \delta_k \dotsb \delta_{\var(w)} [\![w]\!]
            \end{multline*}
            and if $\sigma(v_j) = c$:
            \begin{multline*}
              c \cdot \texttt{AnswerChallenge($r, \sigma, k$)} + (1 - c) \cdot \texttt{AnswerChallenge($l, \sigma, k$)} \\
              = \Pev\sigma \delta_k \dotsb \delta_{\var(w)} [\![w]\!]
            \end{multline*}
            which are the respective results returned by the algorithm in this case.
      \item Case $\var(w) < k$: We make use of the fact that all variables
            occurring in $\Eval{u}$ and $\Eval{v}$ are not in $\in [\var(w),
            \dots, k]$ due to the ordering constraint.
            \begin{align*}
              \texttt{AnswerChallenge($u, \sigma, k$)}                                                           & = \Pev\sigma  [\![u]\!]                                 &  & (\text{IH}) \\
              \texttt{AnswerChallenge($v, \sigma, k$)}                                                           & = \Pev\sigma  [\![v]\!]                                 &  & (\text{IH}) \\
              [\![\texttt{AnswerChallenge($u, \sigma, k$)} \hspace{5em} \\ 
              \circledast\ \texttt{AnswerChallenge($v, \sigma, k$)}]\!] & = \Pev\sigma [\![u]\!] \circledast \Pev\sigma [\![v]\!]                  \\
                                                                                                              & = \Pev\sigma[\![u \circledast v]\!]
            \end{align*}
            The conclusion follows, since the algorithm returns $\Pev\sigma
              [\![u]\!] \circledast \Pev\sigma [\![v]\!]$.
    \end{itemize}

  \paragraph{Step $w = \langle v_j, l, r \rangle$}: We have
  \begin{align*}
    \texttt{AnswerChallenge($l, \sigma, k$)} & = \Pev\sigma \delta_k \dotsb \delta_{\var(l)} [\![l]\!] &  & (\text{IH}) \\
    \texttt{AnswerChallenge($r, \sigma, k$)} & = \Pev\sigma \delta_k \dotsb \delta_{\var(r)} [\![r]\!] &  & (\text{IH})
  \end{align*}
  Also, $\delta_j [\![w]\!] = [\![w]\!]$ since $w$ is already in
  standard form and we can again rewrite the return value to match our proposition: If $v_j
    \notin \sigma$:
    \begin{multline*}
    v_j \cdot \texttt{AnswerChallenge($r, \sigma, k$)} + (1 - v_j) \cdot \texttt{AnswerChallenge($l, \sigma, k$)} \\
    = \Pev\sigma \delta_k \dotsb \delta_{\var(w)} [\![w]\!]
    \end{multline*}
  and if $\sigma(v_j) = c$:
  \begin{multline*}
    c \cdot \texttt{AnswerChallenge($r, \sigma, k$)} + (1 - c) \cdot \texttt{AnswerChallenge($l, \sigma, k$)} \\
    = \Pev\sigma \delta_k \dotsb \delta_{\var(w)} [\![w]\!]
  \end{multline*}
  which are the respective results returned by the algorithm in this case.
\end{proof}

Unlike partial evaluation challenges, \procname{ChallengeDistinct}
(Alg.~\ref{alg:distinct}) is only called by \prot\ on degree-reduced \gbcshort\ nodes.
Because degree-reduced nodes are represented using final BDDs, the link field of
our eBDDs is unused, and Prover generates a distinct assignment without the need
for an additional parameter $k$. It is easy to see that
\procname{ChallengeDistinct($w, r$)} runs in time $O(\min(|w|, |r|))$.

\begin{algorithm}
  \caption{\texttt{ChallengeDistinct($w, r$)}}
  \label{alg:distinct}
  \begin{algorithmic}
    \REQUIRE $w,\ r$ \COMMENT{Distinct, standard BDDs}
    \ENSURE Assignment $\sigma$ on which $[\![w]\!] \neq [\![r]\!]$
    \algrule
    \IF {$w \in \{0, 1\}$ \OR $r \in \{0, 1\}$}
    \RETURN $\sigma$ \COMMENT{Terminal case, we can return any total assignment}
    \ENDIF
    \STATE $\langle v_i, l_\varphi, r_\varphi\rangle = w$
    \STATE $\langle v_j, l_\psi, r_\psi\rangle = r$
    \IF {$v_i < v_j$}
    \STATE $l_\varphi \leftarrow w,\ r_\varphi \leftarrow w$
    \ELSIF {$v_j < v_i$}
    \STATE $l_\psi \leftarrow r,\ r_\psi \leftarrow r$
    \ENDIF
    \IF {$l_\varphi \neq l_\psi$}
    \STATE $\sigma \leftarrow \texttt{ChallengeDistinct($l_\varphi, l_\psi$)}$
    \RETURN $\sigma[\max(v_i, v_j) \rightarrow 0]$
    \ELSE
    \STATE $\sigma \leftarrow \texttt{ChallengeDistinct($r_\varphi, r_\psi$)}$
    \RETURN $\sigma[\max(v_i, v_j) \rightarrow 1]$
    \ENDIF
  \end{algorithmic}
\end{algorithm}

\section{Proving Proposition~\ref{prop:computebdd}}
\label{apx:proposition1}

Having shown that \procname{ApplyEBDD} computes a correct eBDD in
Appendix~\ref{apx:newname} and that Prover can use it to compute the
arithmetization of each gate in $\conv(\varphi)$ in
Appendix~\ref{apx:challengeprover}, we can now prove
Proposition~\ref{prop:computebdd} of the main text.

\spnewtheorem*{P1}{Proposition~\ref{prop:computebdd}}{\bfseries \upshape}{\itshape}
\begin{P1} Let $\psi_1,\psi_2$ denote nodes of
  $\conv(\varphi)$ and $u_1,u_2$ BDDs with $\Eval{u_i}=\Eval{\psi_i}$,
  $i\in\{1,2\}$. Then $\texttt{ApplyEBDD}(u_1, u_2, \circledast)$ satisfies
  $\Eval{w_0}=\Eval{\psi_1\circledast\psi_2}$ and
  $\Eval{w_{i+1}}=\delta_{x_{n-i}}\Eval{w_{i}}$ for every $0 \leq i \leq n-1$;
  moreover, $w_n$ is a BDD with $w_n=\texttt{Apply}(u_1,u_2, \circledast)$.
  Finally, the algorithm runs in time $O(T)$, where $T$ is the time taken by
  $\texttt{Apply($u_1, u_2, \circledast$)}$.
\end{P1}
\begin{proof}
 From Lemma~\ref{lemma:computeebdds} and
  Lemma~\ref{lemma:evaluateebddcorrectness} we have that the series of eBDDs
  $\Eval{w_0}=\Eval{\psi_1\circledast\psi_2}$ and
  $\Eval{w_{i+1}}=\delta_{x_{n-i}}\Eval{w_{i}}$ for every $0 \leq i \leq n-1$
  are computed by Prover via \procname{ApplyEBDD($u_1, u_2, \circledast$)}.
  Further, Lemma~\ref{lemma:complexity} bounds the algorithm's runtime to $O(T)$.
\end{proof}

\section{\protGC: a bottom-up version of \prot}
\label{apx:gccertify}

We describe \procname{\protGC($\varphi, \mathcal{C}$)}, a bottom-up modification
of \prot\ to allow Prover to remove intermediate eBDD nodes as Solver does.
\protGC\ essentially checks all claims about circuit $\varphi$ and then
propagates them `upwards' to its parent nodes, allowing $\varphi$ to be
removed. This means that \protGC\ is invoked multiple times in a sequence, until
no claims remain, while the top-down \prot\ is invoked only once at the end of
model checking. To accommodate garbage collection, Verifier is changed as
follows:
\begin{enumerate}
  \item At the first invocation of \texttt{\protGC($\varphi, \mathcal{C}$)},
        Verifier creates a random assignment $\sigma$ for all $n$ variables.
  \item Each time \prot\  samples a random value for a variable $v_k \R
        \F$, \protGC\ instead uses the random assignment $v_k \leftarrow
        \sigma(v_k)$ computed at the beginning.
  \item Verifier adds the claim $\Pev\sigma \Eval{\psi} = k$ to $\Claimset$ for
        each node $\psi \in \varphi$. Checking this claim allows Verifier to
        propagate claims to $\varphi$'s parents in the next invocation: Verifier
        replaces $\psi$ by a \emph{new GBC node type} $\varepsilon_k^\psi$ after
        running \protGC. Note that this node type does not have any children, so
        unreachable nodes can be garbage collected.
  \item In subsequent invocations, when checking a claim $\Pev{\sigma'}
        [\![\varepsilon_k^\psi]\!] = k'$ (corresponding to a previously
        removed node), Verifier rejects iff.\ $\sigma' \neq \sigma$ or $k
        \neq k'$.
\end{enumerate}

The main text explains how, if a dishonest Prover can `remember' old challenges,
then \protGC\ would no longer be sound. However, under the reasonable assumption
that Prover does not store previous communication with Verifier, \protGC\ is
sound and complete. The proof requires an additional Lemma that allows for sequential
execution of the protocol (\prot\ is run once, and its correctness proof only concerns
a single execution).
\begin{lemma}
  For a series of interactive proofs $\procname{\protGC}_j(\varphi_j,
    \mathcal{C}_j)$ for $j \in 0 \dotsc i$ and any Prover, assuming that Prover
    cannot remember previous challenges, i.e.\ it models an oracle, if any
    $\mathcal{C}_k$ contains a wrong \emph{assignment} claim about the $n$-{\gbcshort}
    $\varphi_k$, \procname{\protGC$_k$($\varphi_k, \mathcal{C}_k$)} will
    accept with probability at most $\left({4n \sum_{c = 0}^{k}
    \left|\varphi_c\right|}\right)/{\F}$. If $\mathcal{C}_k$ contains
    only wrong claims of non-assignment type,
    \procname{\protGC$_k$($\varphi_k, \mathcal{C}_k$)} will accept with
    probability at most $\left({4n \sum_{c = 0}^{k} \left|\varphi_c\right| +
    n}\right)/{\F}$ (soundness).
    
    If all claims in $\mathcal{C}_k$ are correct, it will accept with
    probability $1$ given an honest Prover (completeness).
    \label{lemma:protgccorrectness}
\end{lemma}

To prove the above Lemma, we refer to the proof of \prot\ (\protGC\ only differs
by minor modifications). Remember that Lemma~\ref{lemma:certifyassignments}
bounds the soundness error of \procname{CertifyAssignments} to at most
$\left(4n|\varphi|\right) / \F$, and it can only reach $\left(4n|\varphi| +
n\right) / \F$ due to $\procname{Normalize}$ if there are only false
\emph{equivalence} claims (Lemma~\ref{lemma:normalizecehp}).
\begin{proof}
  We prove soundness and completeness
  for each invocation $\procname{\protGC}_j,\ j \in [0, i]$ via induction on
  $i$.

  \paragraph{Case $i = 0$:}

  We begin by proving soundness, i.e.\ if a claim in $\mathcal{C}_0$ is false,
  then \procname{\protGC$_0$($\varphi_0, \mathcal{C}_0$)} accepts with probability
  at most $4n|\varphi|+n/\F$. The original proof requires new reasoning for the
  following case: In any round in which a random value $v_k \R \F$
  was sampled in \prot, we now instead use the value sampled at the beginning of
  the protocol $v_k \leftarrow \sigma_R(v_k)$.

  Note that because it is the first invocation, Verifier samples a random
  assignment $\sigma_R$ and the circuit $\varphi_0$ does not include any GBC
  node $\varepsilon_k^\psi$. For a polynomial $p(v_k)$ in \prot\ (sent by
  Prover), the probability that a value randomly sampled by Verifier turned a
  false claim into a true claim is at most $2 / \F$ due to the
  Schwartz-Zippel lemma. In \protGC, this probability is also
  bound by $2 / \F$ due to the Schwartz-Zippel lemma, which we can use
  because the polynomial sent by Prover $p(v_k)$ and $\sigma_R(v_k)$ are
  independent (because Prover cannot `remember' the interaction, i.e.\ Prover is
  an oracle).

  Proving perfect completeness for an honest Prover requires no adjustment to
  the original proof.

  \paragraph{Case $i = j + 1$:}
  By the IH, we only need to bound the error probability for the next invocation
  $j + 1$. We again begin by proving soundness. The proof requires handling the
  new case of checking a false claim of the form $\Pev{\sigma'}
  \Eval{\varepsilon_k^\psi}$: Assume Verifier is checking false claim
  $\Pev{\sigma'} \Eval{\varepsilon_k^\psi} = k'$ on a node. Further Assume that
  $\sigma' = \sigma_R$ and $k = k'$, as otherwise Verifier rejects. Because
  $\Pev{\sigma_R} [\![\varepsilon_k^\psi]\!] = k$ is false, we have that
  $\Pev{\sigma_R} \Eval{\psi} \neq k$. However, the claim $\Pev{\sigma_R}
  \Eval{\psi} = k$ was verified previously by Verifier in invocation
  \procname{\protGC$_c$($\phi_c, \mathcal{C}_c$)} for $c \leq j$ before
  replacing $\psi$ by $\varepsilon_k^\psi$, which we can assume to have accepted
  with probability at most $4n \sum_{t = 0}^{c}|\varphi_t| / \F$ by the
  induction hypothesis.
  
  Now we can prove that $\protGC_i(\varphi_i, C_i)$ rejects with probability at
  least $1 - (4n|\varphi_i| + n)/\F$. In order to falsely accept, Verifier must
  either have replaced a false claim by a true claim as in the original protocol
  with prob.\ at most $\left(4n|\varphi_i| + n /\right) \F$ or all of the false
  claims about dead nodes were wrongly accepted with prob.\ at most $\left(4n
  \sum_{t = 0}^{j}|\varphi_t|\right) / \F$. By the union bound, this gives us an
  upper bound on falsely accepting of $\left( 4n \sum_{c = 0}^{i} |\varphi_c| +
  n \right) / \F$

  To show completeness for the new node type $\varepsilon$, assume that claim
  $\Pev\sigma [\![\varepsilon_k^\psi]\!] = k'$ is correct. The only possibility
  for Verifier to reject said claim is if $\sigma \neq \sigma_R$ or $k \neq k'$.
  We first show that $\sigma = \sigma_R$ when the claim gets propagated to node
  $\varepsilon_K^\psi$.

  Recall from the description of \protGC\ that Verifier added claims
  $\Pev{\sigma_R} \Eval{\phi} = k$ to each node $\phi$ in $\varphi_i$. If
  $\sigma \neq \sigma_R$ or $k \neq k'$, then Verifier will use \procname{Merge}
  at the beginning of visiting node $\varepsilon_k^\psi$ to generate a single
  assignment claim. Because every variable assignment is taken from $\sigma_R$
  in \procname{Merge}, the resulting assignment must be $\sigma_R$. By
  assumption that the claim is correct, $k = k'$. Therefore, Verifier accepts
  with probability $1$.

\end{proof}

\spnewtheorem*{L2}{Lemma~\ref{lemma:eccertifycorrectness}}{\bfseries \upshape}{\itshape}

Lemma~\ref{lemma:protgccorrectness} proves Lemma~\ref{lemma:eccertifycorrectness}
by noting that, with garbage collection, \checker\ performs the series of
certifications $\protGC_j(\varphi_j, \Claimset_j)$ for $j \in [0, i]$, and,
without garbage collection, \checker\ performs one certification $\prot(\varphi,
\Claimset)$ with $\bigcup \varphi_j = \varphi$ and $\varphi_j, \varphi_k$
disjoint.

\begin{L2}[\protGC] {[Soundness/Completeness]}
   If $\mathcal{C}$ contains a false claim about an {\gbcshort} $\conv(\varphi)$ with
  $n$ variables, then Verifier accepts with probability at most $\left({4n
  |\varphi| + n}\right)/{\F}$ for any Prover that acts as an oracle. If
  all claims in $\mathcal{C}$ are true, Verifier accepts with probability $1$
  for the honest prover (completeness).
\end{L2}

\end{document}